\documentclass[11pt]{article}
\usepackage{amsfonts}
\usepackage[letterpaper,margin=1.2in]{geometry}

\newtheorem{theorem}{Theorem}[section]
\newtheorem{definition}{Definition}[section]
\newtheorem{lemma}{Lemma}[section]

\newtheorem{proposition}{Proposition}[section]

\newenvironment{proof}{\par \noindent
            {\bf Proof. \hspace{2mm}}}{\hfill$\Box$ \vspace*{3mm}}

\def\bra#1{\langle #1 |}
\def\ket#1{| #1 \rangle}
\def\kb#1#2{\ket{#1}\bra{#2}}
\def\braket#1#2{\bra{#1}{#2}\rangle}

\def\udn#1{#1\in\{0,1\}^n}

\newcommand{\nat}{\mathbb{N}}

\begin{document}
\title{\Large\bf 
%\begin{center}
Non-Interactive Statistically-Hiding Quantum Bit\\
Commitment from Any Quantum One-Way Function}
%\\~
%\end{center}
%\\ (Extended Abstract)}
\author{Takeshi Koshiba \and Takanori Odaira}
%\institute{Graduate School of Science and Engineering, Saitama University}
\date{Graduate School of Science and Engineering, Saitama University}
%\\[1em]
%\today}
\maketitle

\begin{abstract}
We provide a non-interactive quantum bit commitment scheme which 
has statistically-hiding and computationally-binding properties from any
quantum one-way function. Our protocol is basically a parallel composition
of the previous non-interactive quantum bit commitment schemes
(based on quantum one-way permutations,
due to Dumais, Mayers and Salvail (EUROCRYPT 2000))
with pairwise independent hash functions.
To construct our
non-interactive quantum bit commitment scheme from any quantum one-way function,
we follow the procedure below: (i) from Dumais-Mayers-Salvail scheme to a
weakly-hiding and 1-out-of-2 binding commitment (of a parallel variant); 
(ii) from the weakly-hiding and 1-out-of-2 binding commitment to
a strongly-hiding and 1-out-of-2 binding commitment; (iii)
from the strongly-hiding and 1-out-of-2 binding commitment to 
a normal statistically-hiding commitment.
In the classical case, statistically-hiding bit commitment scheme 
(by Haitner, Nguyen, Ong, Reingold and Vadhan (SIAM J. Comput., Vol.39, 2009))
is also constructible from any one-way function. While the classical statistically-hiding 
bit commitment has large round complexity, our quantum scheme is non-interactive, 
which is advantageous over the classical schemes. 
A main technical contribution is to provide a quantum analogue of
the new interactive hashing theorem, due to Haitner and Reingold (CCC 2007). 
Moreover, the parallel composition enables us to
simplify the security analysis drastically.
\end{abstract}

\begin{quote}
\end{quote}

\section{Introduction}
A bit commitment is a fundamental cryptographic protocol between two parties.
The protocol consists of two phases: commit phase and reveal phase. 
In the commit phase, the sender, say Alice, has a bit $b$ in her private space 
and she wants to commit $b$ to the receiver, say Bob. 
They exchange messages and at the end of the commit phase
Bob gets some information that represents $b$. In the reveal phase, 
Alice confides $b$ to Bob by exchanging messages. At the end of the reveal phase,
Bob judges whether the information gotten in the reveal phase really represents
$b$ or not.
Basically, there are three requirements for secure bit commitment:
the correctness, the hiding property and the binding property. 
The correctness guarantees that if both parties are honest then, 
for any bit $b\in\{0,1\}$ Alice has, Bob accepts with certainty.
The hiding property guarantees that (cheating) Bob cannot reveal 
the committed bit during the commit phase. The binding property guarantees that
(cheating) Alice cannot commit her bit $b$ such that Alice maliciously
reveal $b\oplus 1$ as her committed bit but Bob accepts.

In the classical case, a simple argument shows the impossibility
of bit commitment with the hiding and the binding properties both statistical.
Thus, either hiding or binding must be computational.
A construction of statistically-binding scheme 
from any pseudorandom generator
was given by Naor \cite{Naor91}. Since the existence of one-way functions
is equivalent to that of pseudorandom generators \cite{HILL99},
the statistically-binding scheme can be based on any one-way function.
A construction of statistically-hiding scheme ({\sf NOVY} scheme)
from one-way permutation was given by Naor, Ostrovsky, Venkatesan
and Yung \cite{NOVY98}. After that, the assumption of the existence
of one-way permutation was relaxed to that of approximable-preimage-size
one-way function \cite{HHKKMS05}. Finally, Haitner and Reingold \cite{HR07b}
showed that a statistically-hiding scheme ({\sf HNORV} scheme \cite{HNORV09}) can 
be based on any one-way function.

Since statistically-binding (resp., statistically-hiding) 
bit commitment schemes are used as building block for zero-knowledge
proof (resp., zero-knowledge argument) systems \cite{GMW91,BCC88},
it is desirable to be efficient from several viewpoints
(e.g., the total size of messages exchanged during the protocol, 
or the round number of communications in the protocol). 
In general, the round complexity of statistically-hiding schemes is large
(see, e.g., \cite{HR07a,HHRS07}).

Let us move on the quantum case. After the unconditionally security
of the BB84 quantum key distribution protocol \cite{BB84} was shown,
the possibility of unconditionally secure quantum bit commitments had
been investigated. Unfortunately, the impossibility of unconditionally
secure quantum bit commitment was shown \cite{LC97,Mayers97}.
After that, some relaxations such as quantum string commitment \cite{Kent03}
or cheat-sensitive quantum bit commitment \cite{ATVY00,HK04,BCHLW08} have
been studied. 

In this paper, we take the computational approach as in the classical case. 
Along this line, 
Dumais, Mayers and Salvail \cite{DMS00} showed a construction
of perfectly-hiding quantum bit commitment scheme ({\sf DMS} scheme)
based on quantum one-way permutation. 
The non-interactivity in {\sf DMS} scheme is advantageous over
the classical statistically-hiding bit commitments. Unfortunately, we have
not found any candidate of quantum one-way permutation, because known
candidates for classical one-way permutation are no longer one-way
in the quantum setting due to Shor's algorithm \cite{Shor97}. 
Koshiba and Odaira \cite{KO09} observed that the binding property
of {\sf DMS} scheme holds for any quantum one-way functions
and showed that
any approximable-preimage-size quantum one-way function
suffices for the statistical hiding property.

In this paper, we further generalize 
statistically-hiding quantum bit commitment schemes in \cite{DMS00,KO09} and
show that a statistically-hiding quantum
bit commitment is constructible from any {\em general\/} quantum one-way 
function without losing the non-interactivity.
%\footnote{We note that if we allow interaction then
%Cr{\'e}peau, L{\'e}gar{\'e} and Salvail \cite{CLS01} 
%(together with \cite{Naor91,HILL99}) 
%show the existence of
%statistically-hiding and computationally-binding 
%quantum bit commitment schemes based on any quantum one-way function.}.
%by 
%of the round complexity $O(n^2)$ based on
%any quantum one-way function, where $n$ is the security parameter.}.
We basically follow the steps of the proof in \cite{HNORV09}. Thus,
we remark the similarity and differences.
%Our construction is described as follows:
\begin{itemize}\itemsep=0pt\parsep=0pt
\item 
As in {\sf HNORV} scheme \cite{HNORV09},
we consider to construct a 1-out-of-2 binding commitment scheme
based on any
%a regular 
quantum one-way function
%with the unknown preimage-size
as an intermediate scheme.
Our 1-out-of-2 binding commitment scheme executes in parallel
two commitment schemes where one of the two commitment schemes 
satisfies the binding property and the other does not have to satisfy
the binding property. 
%%%% I added the following (2011/2/12)
Note that any adversary for the classical 1-out-of-2 binding commitment
(of the serial composition) cannot see the second commitment just after
getting the first commitment. But, in our case, the adversary can get both
the first and the second commitments, which may be correlated. Thus, we have
to cope with the adversary that can have more information.
%%%%
One of important technical tools in \cite{HNORV09} is so-called
``new interactive hashing theorem'' \cite{HR07a}. 
We provides a quantum analogue of the new interactive hashing theorem.
\item 
%Next, we show that the above construction still holds for the
%case of general quantum one-way functions.
%\item 
Since the resulting 1-out-of-2 binding commitment scheme 
satisfies the hiding property only in a weak sense, 
some hiding amplification technique is applied to yield a
1-out-of-2 binding commitment scheme with the hiding property in a strong sense.
To this end, we just consider the repetitional use of quantum one-way function
and show that the simple repetition works for the hiding amplification.
In \cite{HNORV09}, an amplification procedure is recursively iterated
and an iterative analysis is made. 
Due to the parallel composition, we can drastically simplify the security analysis
of the hiding amplification.
\item Finally, we construct a (normal) statistically-hiding quantum bit commitment
from the 1-out-of-2 binding commitment scheme. Unlike {\sf HNORV} scheme, 
we do not use, in this step, the technique of universal one-way hash functions,
which requires interactions.
\end{itemize}

\medskip
\noindent
{\it Remark.} In the quantum setting, there are several definitions
for the binding property of commitment schemes. In \cite{DFRSS07}, 
a satisfactory definition is given. Nonetheless, we adopt a weaker
definition as in \cite{DMS00} and construct a non-interactive quantum
bit commitment scheme based on the weak definition. It seems much
more difficult to prove the security of our construction according
to the definition in \cite{DFRSS07} by using our techniques.
I believe that the weak definition is sufficient for some applications.
Actually, a construction of quantum oblivious transfer from a quantum string 
commitment (of a special type) with a similar weak binding condition was 
given in \cite{CDMS04}.

\section{Preliminaries}
\subsection{Notations and Conventions}
We denote the $m$-dimensional Hilbert space by ${\sf H}_m$.
Let $\{\ket{0},\ket{1}\}$ denote the computational basis for ${\sf H}_2$.
When the context requires, we write $\ket{b}_{+}$ to denote $\ket{b}$
in the computational basis. Let $\{\ket{0}_{\times},\ket{1}_{\times}\}$ denote
the diagonal basis, where $\ket{0}_{\times}=\frac{1}{\sqrt{2}}(\ket{0}+\ket{1})$
and $\ket{1}_{\times}=\frac{1}{\sqrt{2}}(\ket{0}-\ket{1})$.
For any $x=x_1x_2\cdots x_n\in \{0,1\}^n$ and $\theta\in\{+,\times\}$,
$\ket{x}_\theta$ denotes the state $\otimes_{i=1}^n \ket{x_i}_\theta$.
We denote $\ket{0}\otimes \cdots \otimes \ket{0}$ by $\ket{\bf 0}$.
For projections, we denote ${\cal P}^0_{+}=\kb{0}{0}$, ${\cal P}^1_{+}=\kb{1}{1}$, 
${\cal P}^0_{\times}=\ket{0}_{\times}\bra{0}$, and
${\cal P}^1_{\times}=\ket{1}_{\times}\bra{1}$.
For any $\udn{x}$, we denote ${\cal P}^x_{+}=\otimes_{i=1}^n {\cal P}^{x_i}_{+}$
and ${\cal P}^x_{\times}=\otimes_{i=1}^n {\cal P}^{x_i}_{\times}$. For the sake of
simplicity, we also write ${\cal P}^x$ instead of ${\cal P}^x_{+}$.
We define $\theta(0) = +$ and $\theta(1)=\times$. Thus, for any $w\in \{0,1\}$,
$\{{\cal P}^x_{\theta(w)}\}_{\udn{x}}$ is the von Neumann measurement.
For density matrices $\sigma$ and $\rho$, we define 
$\delta(\sigma,\rho)\stackrel{\rm def}{=}\|\sigma-\rho\|_1$, 
where $\|A\|_1 = \frac{1}{2}{\rm tr}\sqrt{A^\dag A}$. 
For two classical random variables $X$ and $Y$, there exists the corresponding
density matrices $\rho_X$ and $\rho_Y$. Since $\delta(\rho_X,\rho_Y)$ also
represents the variation distance (a.k.a. statistical distance) between $X$ and $Y$,
we sometimes write $\delta(X,Y)$ instead of $\delta(\rho_X,\rho_Y)$.
We denote the min-entropy of a random variable $X$ by ${\bf H}_\infty(X)$
and the Renyi entropy (of order 2) by ${\bf H}_2(X)$.
%For a classical random variable $X$, we denote, by ${\rm cp}(X)$,
%the collision probability $\sum_{x\in {\rm supp}(X)} (\Pr[X=x])^2$.
%We note that ${\bf H}_2(X)=-\log {\rm cp}(X)$.
%We define ${\rm cp}^{1/2}(X)=\sqrt{{\rm cp}(X)}$.
%
We denote the uniform distribution over $\{0,1\}^n$ by $U_n$. For a set $A$,
we sometimes use the same symbol to denote the uniform distribution over the set $A$.
A function $\nu: \mathbb{N}\rightarrow \mathbb{R}$ is {\em negligible\/} if
for every polynomial $p$ there exists $n_0\in\nat$ such for all $n\ge n_0$,
$\nu(n)<1/p(n)$. We denote a set of integers $\{i\in\mathbb{N}: n_1\le i\le n_2\}$ by
$[n_1,n_2]$.

\subsection{Quantum One-Way Functions}
In order to give definitions of quantum one-way functions, we have to decide
a model of quantum computation. In this paper, we consider (uniform or non-uniform) quantum circuit family. 
As a universal quantum gate set, we take the
controlled-NOT, the one-qubit Hadamard gate, and arbitrary one-qubit
non-trivial rotation gate. The computational complexity of a circuit $\cal C$
is measured by the number of elementary gates (in the universal gate set)
contained in $\cal C$ and denoted by ${\it size}({\cal C})$.
For any circuit family ${\cal C}=\{{\cal C}_n\}_{n\in\nat}$,
if ${\it size}({\cal C}_n)$ is bounded by $p(n)$ for some polynomial $p$,
${\cal C}$ is called p-size circuit family.

Let $f=\{f_n:\{0,1\}^n \rightarrow \{0,1\}^{\ell (n)}\}_{n\in\nat}$
be a function family. To compute $f$, we need a circuit family 
$\{{\cal C}_n\}_{n\in\nat}$ where ${\cal C}_n$ is a circuit on $m(n)\ge \ell(n)$
qubits. To compute $f_n(x)$ for $\udn{x}$, we apply 
${\cal C}_n$ to $\ket{x}\otimes \ket{0}^{\otimes m(n)-n}$.
The output of ${\cal C}_n$ is obtained by the von Neumann measurement
in the computational basis on $\ell(n)$ qubits.

\begin{definition}\rm
A function family $f=\{f_n:\{0,1\}^n \rightarrow \{0,1\}^{\ell (n)}\}_{n\in\nat}$
is $s(n)$-{\em secure quantum one-way\/} if
\begin{itemize}
\item there exists a $p$-size circuit family ${\cal C}=\{{\cal C}_n\}_{n\in\nat}$
such that, for all $n\ge 1$ and all $\udn{x}$,
${\cal C}_n(\ket{x}\otimes \ket{{\bf 0}}) = f_n(x)$ with certainty;
\item for every $p$-size circuit family ${\cal B}=\{{\cal B}_n\}_{n\in\nat}$
and for sufficiently large $n$,
\[ 
\Pr[f_n({\cal B}_n(f_n(U_n)))=f_n(U_n)]<1/s(n). 
\]
\end{itemize}
If $f$ is $p(n)$-secure quantum one-way for every polynomial $p$, then $f$ is said to
be {\em p-secure}.
\end{definition}

Quantum one-way function $f$ is said to be $r(n)$-regular if
for any $y\in {\rm supp}(f(U_n))$, $|\{x\in \{0,1\}^n : f_n(x)=y\}|=2^{r(n)}$.
Without loss of generality, we can consider quantum one-way functions that are
length-preserving, that is, $\ell(n)=n$, because general quantum one-way functions
can be converted into ones that are length-preserving.

\subsection{Quantum Bit Commitment}
In a non-interactive quantum bit commitment scheme, honest Alice with her bit $w\in\{0,1\}$
starts with a system ${\sf H}_{\rm all}={\sf H}_{\rm keep}\otimes
{\sf H}_{\rm open}\otimes {\sf H}_{\rm commit}$ in the initial state
$\ket{\bf 0}$, executes a quantum circuit ${\cal C}_{n,w}$ on $\ket{\bf 0}$
returning the final state $\ket{\psi_w}\in {\sf H}_{\rm all}$ and
finally sends the subsystem ${\sf H}_{\rm commit}$ to Bob in the reduced state
$\rho_B(w)={\rm tr}_A(\kb{\psi_w}{\psi_w})$, where Alice's Hilbert space
is ${\sf H}_A = {\sf H}_{\rm keep}\otimes {\sf H}_{\rm open}$.
For $w\in\{0,1\}$, we call $\rho_B(w)$ {\em $w$-commitment state}.
Once the system ${\sf H}_{\rm commit}$ (or, $w$-commitment state) is sent to Bob, 
Alice has only access
to $\rho_A(w)={\rm tr}_B(\kb{\psi_w}{\psi_w})$, where Bob's Hilbert space
is ${\sf H}_B={\sf H}_{\rm commit}$. To reveal the commitment, Alice needs
only to send the system ${\sf H}_{\rm open}$ together with $w$. Bob then checks
the value of $w$ by measuring the system ${\sf H}_{\rm open}\otimes {\sf H}_{\rm commit}$ with some measurement that is fixed by the protocol in view of $w$.
Bob obtains $w=0$, $w=1$, or $w=\bot$ when the value of $w$ is rejected.

Cheating Alice must start with the state $\ket{\bf 0}$ of some system
${\sf H}_{\rm all}={\sf H}_{\rm extra}\otimes 
{\sf H}_A \otimes {\sf H}_{\rm commit}$. A quantum circuit ${\cal D}_n$
that acts on ${\sf H}_{\rm all}$ is executed to obtain a state $\ket{\psi}$
and the subsystem ${\sf H}_{\rm commit}$ is sent to Bob. Later,
any quantum circuit ${\cal O}_n$ which acts on ${\sf H}_{\rm extra}\otimes
{\sf H}_{\rm keep}\otimes {\sf H}_{\rm open}$ can be executed before
sending the subsystem ${\sf H}_{\rm open}$ to Bob. The important
quantum circuits which act on ${\sf H}_{\rm extra}\otimes
{\sf H}_{\rm keep}\otimes {\sf H}_{\rm open}$ are the quantum circuits
${\cal O}_{n,0}$ (resp., ${\cal O}_{n,1}$) which maximizes the probability
that bit $w=0$ (resp., $w=1$) is revealed with success.
Therefore, any attack can be modeled by triplets of quantum circuits
$\{({\cal D}_n, {\cal O}_{n,0}, {\cal O}_{n,1})\}_{n\in\nat}$.

Let $b_0(n)$ (resp., $b_1(n)$) be the probability that she succeeds
to reveal 0 (resp., 1) using the corresponding optimal circuit ${\cal O}_{n,0}$
(resp., ${\cal O}_{n,1}$). The definition of $b_w(n)$ explicitly
requires that the value of $w$, which cheating Alice tries to open,
is chosen not only before the execution of the measurement on
${\sf H}_{\rm open}\otimes {\sf H}_{\rm commit}$ by Bob but also
before the execution of the circuit ${\cal O}_{n,w}$ by cheating Alice.

In the quantum setting, it is pointed out in \cite{Mayers97} that 
the requirement ``$b_0(n)=0\lor b_1(n)=0$'' for the binding condition 
is too strong. 
Thus, we adopt a weaker condition 
$b(n)\stackrel{\rm def}{=}b_0(n)+b_1(n)-1\le \varepsilon$
where $\varepsilon(n)$ is negligible, which is the same condition
as in \cite{DMS00}.

Since we consider the computational binding, we modify the above discussion
so as to fit the computational setting. Instead of the triplet
$({\cal D}_n, {\cal O}_{n,0}, {\cal O}_{n,1})$, we
consider a pair $({\cal D}_{n,0},{\cal U}_n)$. If we set
${\cal D}_{n,0} = ({\cal O}_{n,0}\otimes {\cal I}_{\rm commit})\cdot {\cal D}_n$,
and ${\cal U}_n = {\cal O}_{n,1}\cdot {\cal O}_{n,0}^\dag$, we can easily
see that the adversary's strategy does not change.
Note that ${\cal D}_{n,0}$ acts in ${\sf H}_{\rm all}$ and
${\cal U}_n$ is restricted to act only in ${\sf H}_{\rm extra}\otimes
{\sf H}_{\rm keep}\otimes {\sf H}_{\rm open}$.

\begin{definition}\rm
A non-interactive quantum bit commitment is $t(n)$-{\em computationally-binding\/}
if, for every a family $\{({\cal D}_{n,0},{\cal U}_n)\}_{n\in\nat}$ 
of p-size circuit pairs, $b(n)$ is bounded by $t(n)$.
If $t(n)$ is negligible in $n$, 
the non-interactive quantum bit commitment is simply said 
to be {\em computationally-binding}.
\end{definition}

\begin{definition}\rm
A non-interactive quantum bit commitment is $t(n)$-{\em statistically-binding\/}
%if $b(n)\le \varepsilon(n)$.
if $b(n)\le t(n)$.
If $t(n)$ is negligible in $n$,
the non-interactive quantum bit commitment is simply said 
to be {\em statistically-hiding}.
\end{definition}
%A non-interactive quantum bit commitment is {\em statistically-hiding\/}
%if $\delta(\rho_B(0), \rho_B(1))$ is negligible.

As mentioned, a satisfactory definition for the binding property of 
quantum bit commitment schemes is given by Damg{\aa}rd, Fehr, Renner,
Salvail and Schaffner \cite{DFRSS07}. Actually, they show that a
variant of {\sf DMS} scheme satisfies the binding condition in \cite{DFRSS07}.
However, it is still unclear whether the inverting quantum one-way permutation
is reducible to violating the binding condition. 
We rather adopt a weaker definition in \cite{DMS00} in order to
benefit from the computational reducibility.

\subsection{Pairwise Independent Hash Functions}
Let $H = \{ H_n \}_{n\in\nat}$ be a sequence of function families,
where each $H_n$ is a family of functions mapping binary strings of
length $\ell(n)$ to strings of length $v(n)$. We say that $H_n$ is a
{\em pairwise independent\/} (a.k.a. strongly 2-universal) {\em hash family\/} 
if for any distinct $x,x'\in \{0,1\}^{\ell(n)}$
and $y,y'\in  \{0,1\}^{v(n)}$, 
$\Pr_{h\leftarrow H_n}[h(x)=y\land h(x')=y'] = 2^{-2v(n)}$. 
(See, e.g., \cite{CW79} for an implementation of pairwise independent hash family.)
%For example, there exists a pairwise independent hash family $H_n$ such that $|H_n|=2^{2n}$.)

One of the useful applications of pairwise independent hash family is smoothing
the min-entropy of given distribution. 
%The following is also known as privacy amplification.
\begin{lemma}\label{lem:lhl}(Leftover Hash Lemma)
Let $V_n$ be a random variable over $\{0,1\}^{\ell(n)}$ 
such that ${\bf H}_\infty(V_n)\ge \lambda_n$
and $H_n$ be a pairwise independent hash family where each $h\in H_n$ maps
strings of length $\ell(n)$ to 
strings of length $\lambda_n -2\log(\varepsilon^{-1})$.
Then, we have $\delta( (H_n,H_n(V_n)), (H_n,U_{v(n)}) )\le \varepsilon$.
\end{lemma}

\section{Base Scheme}
Dumais, Mayers and Salvail \cite{DMS00} gave a non-interactive statistically-hiding
quantum bit commitment based on quantum one-way permutation. Koshiba and Odaira \cite{KO09}
observed that {\sf DMS} scheme still satisfies the computational binding if we replace 
quantum one-way permutation with general quantum one-way function. So, we consider to
use the scheme as an important ingredient of the construction of our non-interactive statistically-hiding
quantum bit commitment based on quantum one-way function.

We briefly review the scheme.
Let $f=\{ f_n : \{0,1\}^n \rightarrow \{0,1\}^{n} \}_{n\in\nat}$
be a function family. The quantum bit commitment scheme takes
the security parameter $n$ and the description of function family $f$
as common inputs.
For given $f$ and the security parameter $n$, Alice and Bob determine $f_n$.
The protocol, called {\sf Base Protocol}, is described in Figure \ref{fig:bp}.
%as follows.

\begin{figure}[htbp]
%\medskip
\hrule
\medskip
\noindent
{\bf Commit Phase}:
\begin{enumerate}\parsep=0pt\itemsep=0pt
\item Alice with her bit $w$
first chooses $x\in \{0,1\}^n$ uniformly and computes $y=f_n(x)$.
\item Next, Alice sends the quantum state 
$\ket{f_n(x)}_{\theta(w)}\in {\sf H}_{\rm commit}$ to Bob.
\item
Bob then stores the received quantum state until Reveal Phase.
\end{enumerate}
{\bf Reveal Phase}:
\begin{enumerate}\parsep=0pt\itemsep=0pt
\item Alice first announces $w$ and $x$ to Bob.
\item Next, Bob measures $\rho_B$ with measurement
$\{P^y_{\theta(w)}\}_{y\in {\it range}(f_n)}$ and obtains the classical
output $y'\in {\it range}(f_n)$. 
\item Lastly, Bob accepts if and only if $y'=f_n(x)$.
\end{enumerate}
\hrule
\caption{{\sf Base Protocol}}\label{fig:bp}
\end{figure}
%\medskip

\begin{proposition}{\rm (Implicit in \cite{DMS00} and explicit in \cite{KO09})}\label{prop:hide}
Let $f=\{f_n:\{0,1\}^n \rightarrow \{0,1\}^{n}\}_{n\in\nat}$ be
a family of (not necessarily quantum one-way) functions
such that $\delta(f(U_n),U_n')$ is negligible in $n$.
Then, {\sf Base Protocol} is statistically hiding.
\end{proposition}

\begin{proposition}{\rm (Implicit in \cite{DMS00} and explicit in \cite{KO09})}\label{prop:bind}
Let $f=\{f_n:\{0,1\}^n \rightarrow \{0,1\}^{n}\}_{n\in\nat}$ be
an $s(n)$-secure quantum one-way function family. Then {\sf Base Protocol} is 
$O(1/\sqrt{s(n)})$-computationally binding.
\end{proposition}

In this paper, we do not use directly the above properties. We need to generalize
Proposition \ref{prop:bind}.
In non-interactive commitment protocols, Alice sends a commitment $y$ in Commit Phase 
and a decommitment $x$ in Reveal Phase. To make Bob accept, the pair $(y,x)$ must 
be in some binary relation $R_n$. In case of {\sf Base Protocol},
the binary relation is defined as $R_n=\{(f_n(x),x):x\in\{0,1\}^n\}$. 
We can rephrase the statement of Proposition \ref{prop:bind}
in terms of the binary relation $R_n$. It says that if cheating Alice can output
distinct pairs
$(y,x)$ and $(y',x')$ both in $R_n$ such that 
the probability to reveal 0 with success by using $(x,y)$ is $b_0(n)$,
the probability to reveal 1 with success by using $(x',y')$ is $b_1(n)$,
and $b_0(n)+b_1(n)\ge 1 + \sqrt{s(n)}$, 
%she succeeds to reveal 0
%by using $(y,x)$ with probability $b_0(n)$
%and to reveal 1 by using $(y',x')$ with probability $b_1(n)$, and 
%$b_0(n)+b_1(n)\ge 1 + \sqrt{s(n)}$, 
then there exists an algorithm
that, given $f_n(x)$ as input, outputs $x$ such that $(f_n(x),x) \in R_n$ with
probability $\Omega(s(n))$. 
Since {\sf Base Protocol} is based on quantum one-way function, 
the definition of $R_n$ is quite natural. On the other hand, we may define 
a binary relation as $R'_n=\{(f_n(x),x):x\in W_n\}$ by using some subset
$W_n\subseteq \{0,1\}^n$. 
We discuss a generalization of Proposition \ref{prop:bind} in the next section.

\section{Non-interactive Quantum Hashing Theorem}
The following theorem is a quantum correspondence\footnote{%
Exactly speaking, Theorem \ref{lem:L-bind} corresponds to a special case of the
new interactive hashing theorem in \cite{HR07a} and
the current form suffices for our purpose.
As in \cite{HR07a}, we
can derive a more general form of Non-interactive Quantum Hashing Theorem.}
of the new interactive hashing theorem in \cite{HR07a} and 
it is one of the most technical ingredients in this paper. 

\begin{theorem}\label{lem:L-bind}(Non-interactive Quantum Hashing Theorem)
Let $f=\{f_n:\{0,1\}^n \rightarrow \{0,1\}^{n}\}_{n\in\nat}$ be
an $s(n)$-secure quantum one-way function family.
Suppose that $W_n$ is a subset of $\{0,1\}^n$ and define the binary relation $R'_n$
as $R'_n=\{(f_n(x),x):x\in W_n\}$.
If there exists an algorithm against {\sf Base Protocol} that can output
distinct pairs
$(y,x)$ and $(y',x')$ both in $R_n'$ such that 
the probability to reveal 0 with success by using $(x,y)$ is $b_0(n)$,
the probability to reveal 1 with success by using $(x',y')$ is $b_1(n)$,
and $b_0(n)+b_1(n)\ge 1 + \sqrt{s(n)}$, 
%it succeeds to reveal 0
%by using $(y,x)$ with probability $b_0(n)$
%and to reveal 1 by using $(y',x')$ with probability $b_1(n)$, and 
%$b_0(n)+b_1(n)\ge 1 + \sqrt{s(n)}$, 
then there exists another algorithm
that, given $y''\in f_n(W_n)$ as input, outputs $x''$ such that $(y'',x'') \in R_n'$ with
probability $\Omega(s(n))$, where $y''$ is propotionally selected from $f_n(W_n)$.
\end{theorem}

If $W_n$ is closed to $U_n$ in the statement above, Non-interactive 
Quantum Hashing Theorem can be directly applied to construct an inverter of 
the quantum one-way function as in \cite{DMS00,KO09}.
However, if $W_n$ is far from $U_n$, it is not directly related to the inversion
of the quantum one-way function. In the next section, we discuss how to use it
even in the case where $W_n$ is far from $U_n$.

For the proof of Theorem \ref{lem:L-bind},
we can adapt the proof of Proposition \ref{prop:bind}.
In the original proof in \cite{DMS00} of Proposition \ref{prop:bind}, some 
``test circuits'' are utilized. The existence of test circuits is an obstacle
to the generalization. A careful analysis shows that such test circuits are
redundant. 
%We will give the proof in Appendix \ref{section:L-bind}.

\begin{proof}
We separate the whole system into three parts: the system ${\sf H}_{\rm commit}$
that encodes the functional value, the system ${\sf H}_{\rm open}$ that
encodes inputs to the function, and the system ${\sf H}_{\rm keep}$
is the reminder of the system.

\medskip
\noindent
{\bf Perfect Case:}\\
In the perfect case, we can assume that
an adversary $\{{\cal D}_{n,0},{\cal U}_n\}_{n\in\nat}$
reveals the committed bit in both ways perfectly. 
That is, the states $\ket{\psi_{n,0}}$ (resp., $\ket{\psi_{n,1}}$) 
of the whole system when $w=0$ (resp., $w=1$) will be committed
can be written as follows.
\begin{eqnarray}
\ket{\psi_{n,0}} & = & \sum_{x\in W_n}\ket{\alpha_{0,x}}^{\rm keep}
\otimes \ket{x}^{\rm open} \otimes \ket{f_n(x)}^{\rm commit}_{+} = {\cal D}_{n,0}\ket{\bf 0}\quad\mbox{and} \nonumber\\
\ket{\psi_{n,1}} & = & \sum_{x\in W_n}\ket{\alpha_{1,x}}^{\rm keep}
\otimes \ket{x}^{\rm open} \otimes \ket{f_n(x)}^{\rm commit}_{\times} = {\cal U}_n\ket{\psi_{n,0}},\nonumber
\end{eqnarray}
where $\sum_{x\in W_n}\|\,\ket{\alpha_{0,x}}\|^2 
= \sum_{x\in W_n}\|\,\ket{\alpha_{1,x}}\|^2 = 1$.

Let ${\cal P}_{+}^{u,\rm commit}$ and ${\cal P}_{\times}^{u,\rm commit}$
be the projection operators
${\cal P}_{+}^{u}$ and ${\cal P}_{\times}^{u}$ respectively,
acting in ${\sf H}_{\rm commit}$. We are interested in
properties on the state 
$\ket{\varphi_{n,0}^u}={\cal P}_{\times}^{u,\rm commit}\ket{\psi_{n,0}}$
which plays an important role for the inverter.

Now we consider an algorithm to invert $y\in f_n(W_n)$. Thus, we assume that
$y$ is encoded as input to the inverter in ${\sf H}_{\rm inv}$. 
Before considering the inverter,
we consider properties on the states $\ket{\varphi_{n,0}^u}$ for every $u\in\{0,1\}^{n}$:
\begin{enumerate}
\item $\| \ket{\varphi_{n,0}^u}\|^2 = 2^{-n/2}$;
\item there exists an efficient circuit ${\cal W}_n$
on ${\sf H}_{\rm inv}\otimes {\sf H}_{\rm open}\otimes {\sf H}_{\rm commit}$
which if $u$ is in ${\sf H}_{\rm inv}$, unitarily maps 
$\ket{\psi_{n,0}}$ to $2^{n/2}\ket{\varphi_{n,0}^u}$;
\item ${\cal U}_n\ket{\varphi_{n,0}^u} =
\sum_{z\in f_n^{-1}(u)}
\ket{\alpha_{1,z}}^{\rm keep}\otimes
\ket{z}^{\rm open}\otimes \ket{u}_{\times}^{\rm commit}$.
\end{enumerate}

If the above properties are true, we can consider an inverter as follows.
On input $y$, the inverter generates the state $\ket{\psi_{n,0}}$
by applying ${\cal D}_{n,0}$ to $\ket{\bf 0}$, then applies
${\cal W}_n$ and ${\cal U}_n$ in order, and finally measures ${\sf H}_{\rm open}$
to obtain $z\in f_n^{-1}(y)$.

In what follows, we show each property is true.
First, we show Property 1. We write $\ket{\psi_{n,0}}$ using
the diagonal basis for ${\sf H}_{\rm commit}$, and then we have
\begin{eqnarray*}
\ket{\psi_{n,0}} & = & \sum_{x\in W_n}\ket{\alpha_{0,x}}^{\rm keep}
	\otimes \ket{x}^{\rm open} \otimes \left(
	\sum_{u\in\{0,1\}^{n}} \frac{(-1)^{\langle u, f_n(x)\rangle}}{2^{n/2}}\ket{u}^{\rm commit}_{\times}
\right)\\
& = & 2^{-n/2} \!\!\!\!\sum_{%
\stackrel{\scriptsize \mbox{$u\in\{0,1\}^{n}$}}{x\in W_n}}
	(-1)^{\langle u, f_n(x)\rangle}
\ket{\alpha_{0,x}}^{\rm keep}
	\otimes \ket{x}^{\rm open} 
\otimes \ket{u}_{\times}^{\rm commit}.
\end{eqnarray*}
Since
\begin{eqnarray*}
\ket{\varphi_{n,0}^u} & = & 2^{-n/2} \sum_{x\in W_n}
	(-1)^{\langle u, f_n(x)\rangle}
\ket{\alpha_{0,x}}^{\rm keep}
	\otimes \ket{x}^{\rm open} 
\otimes \ket{u}_{\times}^{\rm commit},%\label{eq:after}
\end{eqnarray*}
Property 1 holds. Next, we consider Property 3. Since the state
$\ket{\psi_{n,1}}$ can be written as
\[
\ket{\psi_{n,1}} = \sum_{u\in f_n(W_n)}
	\left( \sum_{z\in f_n^{-1}(u)} \ket{\alpha_{1,z}}^{\rm keep}
	\otimes \ket{z}^{\rm open} \right)
	\otimes \ket{u}_{\times}^{\rm commit},
\]
it implies that for every $u\in f_n(W_n)$
\begin{eqnarray*}
{\cal U}_n \ket{\varphi_{n,0}^u} & = &
{\cal U}_n {\cal P}_{\times}^{u,\rm commit} \ket{\psi_{n,0}}
= {\cal P}_{\times}^{u,\rm commit} {\cal U}_n \ket{\psi_{n,0}}\\
& = & {\cal P}_{\times}^{u,\rm commit} \ket{\psi_{n,1}} =
	\sum_{z\in f_n^{-1}(u)} \ket{\alpha_{1,z}}^{\rm keep}
	\otimes \ket{z}^{\rm open}
	\otimes \ket{u}_{\times}^{\rm commit}.
\end{eqnarray*}
(Note that ${\cal U}_n$ is restricted to act in 
${\sf H}_{\rm keep}\otimes {\sf H}_{\rm open}$ and thus
${\cal U}_n$ and ${\cal P}_{\times}^{u,\rm commit}$ are commutable.)
Thus, Property 3 holds. Finally, we consider Property 2.
We describe how to implement ${\cal W}_n$ 
mapping from
\[ %\sum_{\udn{x}} 
\ket{u}^{\rm inv}\otimes \ket{x}^{\rm open}\otimes
					\ket{f_n(x)}^{\rm commit}_{+} \]
into
\[ % 2^{-\ell(n)/2} \sum_{\udn{x}} 
(-1)^{\langle u, f_n(x)\rangle}
\ket{u}^{\rm inv}\otimes \ket{x}^{\rm open}\otimes
					\ket{u}^{\rm commit}_{\times} \]
for every $u\in f_n(W_n)$,
% and $x\in W_n$,
which satisfies the requirement.
First we apply the mapping
$\ket{u}^{\rm inv}\otimes\ket{f_n(x)}^{\rm commit}
\mapsto (-1)^{\langle u, f_n(x)\rangle}\ket{u}^{\rm inv}\otimes\ket{f_n(x)}^{\rm commit}$,
which can be efficiently 
implemented by using the Hadamard gate and the controlled-NOT gate.
Secondly, we apply the mapping
$\ket{x}^{\rm open}\otimes\ket{u}^{\rm commit} \mapsto 
	\ket{x}^{\rm open}\otimes\ket{u\oplus f_n(x)}^{\rm commit}$,
which can be implemented by the efficient evaluation circuit of $f_n$.
Thirdly, we apply the mapping
$\ket{y}^{\rm inv}\otimes \ket{u}^{\rm commit} \mapsto
	\ket{y}^{\rm inv}\otimes \ket{y\oplus u}^{\rm commit}$,
which can be efficiently implemented by using the controlled-NOT gate.
Finally, we apply the Hadamard gate to the all qubits in ${\sf H}_{\rm commit}$.
It is easy to verify that the above procedure satisfies the requirement.
Thus, Property 2 holds. 

\medskip
\noindent
{\bf General Case:}\\
In the general case, 
the states $\ket{\tilde{\psi}_{n,0}} = {\cal D}_{n,0}\ket{\bf 0}$ and
$\ket{\tilde{\psi}_{n,1}} = {\cal U}_n\ket{\tilde{\psi}_{n,0}}$ can 
be generally written as 
\begin{eqnarray*}
\ket{\tilde{\psi}_{n,0}} & = & \sum_{x\in \{0,1\}^n,y\in\{0,1\}^{n}}\ket{\alpha_{0,x,y}}^{\rm keep}
\otimes \ket{x}^{\rm open} \otimes \ket{y}^{\rm commit}_{+},\quad\mbox{and} \\
\ket{\tilde{\psi}_{n,1}} & = & \sum_{x\in \{0,1\}^n,y\in\{0,1\}^{n}}\ket{\alpha_{1,x,y}}^{\rm keep}
\otimes \ket{x}^{\rm open} \otimes \ket{y}^{\rm commit}_{\times},
\end{eqnarray*}
where $\sum_{x,y}\|\,\ket{\alpha_{0,x,y}}\|^2 
= \sum_{x,y}\|\,\ket{\alpha_{1,x,y}}\|^2 = 1$.

We assume that $b_0(n)+b_1(n)\ge 1 + 1/p(n)$ for
some polynomial $p$, where
\[
b_0(n)=\sum_{x\in W_n}\|\,\ket{\alpha_{0,x,f_n(x)}}\|^2\quad\mbox{and}\quad
b_1(n)=\sum_{x\in W_n}\|\,\ket{\alpha_{1,x,f_n(x)}}\|^2.
\]
Then we will show that 
the success probability $p_{\rm inv}$
for inverting the underlying quantum one-way function is greater than
$1/4(p(n))^2$.

First, the state $\ket{\tilde{\psi}_{n,0}}$ can be written as follows.
\begin{eqnarray*}
\ket{\tilde{\psi}_{n,0}} & = &
\sum_{x\in W_n} \ket{\alpha_{0,x,f_n(x)}}^{\rm keep}\otimes\ket{x}^{\rm open}
	\otimes\ket{f_n(x)}^{\rm commit}\\
&& + 
\sum_{f_n(x)\ne z~{\rm or}~x\not\in W_n} \ket{\alpha_{0,x,z}}^{\rm keep}\otimes\ket{x}^{\rm open}
	\otimes\ket{z}^{\rm commit}.
\end{eqnarray*}
Remember that the state in the perfect case can be written as
\[ \ket{\psi_{n,0}} = \sum_{x\in W_n} \ket{\alpha_{x,0}}^{\rm keep}
	\otimes \ket{x}^{\rm open} \otimes \ket{f_n(x)}^{\rm commit}. \]
Then we have
\[\ket{\alpha_{0,x}}^{\rm keep} 
	= (b_0(n))^{-1/2}\ket{\alpha_{0,x,f_n(x)}}^{\rm keep}
\quad\mbox{and}\quad
b_0(n)=\sum_{x\in W_n} \|\,\ket{\alpha_{0,x,f_n(x)}}\|^2=|\braket{\psi_{n,0}}{\tilde{\psi}_{n,0}}|^2. \]

On input $y$, the inverter generates the state $\ket{\tilde{\psi}_{n,0}}$
by applying ${\cal D}_{n,0}$ to $\ket{\bf 0}$.
We then apply in order
${\cal W}_n$ and ${\cal U}_n$ to the resulting state
and finally measures ${\sf H}_{\rm open}$ to hopefully obtain $z\in f_n^{-1}(y)$.

We have to estimate the success probability of the inverter.
To this end, we define two projections:
\begin{eqnarray*}
 {\cal P}_0 & \stackrel{\rm def}{=} & \sum_{x\in W_n} 
		{\cal P}^{x,\rm open}\otimes {\cal P}_{+}^{f_n(x),\rm commit}
\quad\mbox{and}\\
{\cal P}_1 & \stackrel{\rm def}{=} & \sum_{x\in W_n} 
		{\cal P}^{x,\rm open}\otimes {\cal P}_{\times}^{f_n(x),\rm commit}.
\end{eqnarray*}
Then 
we have $b_0(n)=\| {\cal P}_0\ket{\tilde{\psi}_{n,0}}\|^2$
and $b_1(n)=\| {\cal P}_1\ket{\tilde{\psi}_{n,1}}\|^2$. Here, we claim that
the success probability $p_{\rm inv}$ satisfies
\[ p_{\rm inv}=\| {\cal P}_1 {\cal U}_n {\cal P}_0 \ket{\tilde{\psi}_{n,0}}\|^2.\]

We will see this claim. As mentioned, %after the application of ${\cal T}_n$,
the state is $\ket{y}^{\rm inv}\otimes\ket{\psi_{n,0}}$ with probability
$\| {\cal P}_0 \ket{\tilde{\psi}_{n,0}}\|^2=b_0(n)$, where $y$ is the input
to the inverter. As we see in the perfect case, ${\cal W}_n$
maps the state $\ket{\psi_{n,0}}$ into 
$2^{n/2}\ket{\varphi_{n,0}^y}=2^{n/2}
{\cal P}_{\times}^{y,\rm commit}\ket{\psi_{n,0}}$.
After that, we apply ${\cal U}_n$ and measure ${\sf H}_{\rm open}$.
Thus, the success probability $p_{\rm inv}(y)$ for input $y$ is written as
\begin{eqnarray*}
p_{\rm inv}(y) &= & b_0(n)2^{n}
	\left\| \left(\sum_{z\in f_n^{-1}(y)} {\cal P}^{z,\rm open}\right) {\cal P}_{\times}^{y,\rm commit}
	{\cal U}_n \ket{\psi_{n,0}}\right\|^2\\
& = &2^{n}\left\|\left(\sum_{z\in f_n^{-1}(y)}{\cal P}^{z,\rm open}\right) {\cal P}_{\times}^{y,\rm commit}
	{\cal U}_n {\cal P}_0 \ket{\tilde{\psi}_{n,0}}\right\|^2.
\end{eqnarray*}
Averaging over all value according to the output distribution of $f_n$, we have
\begin{eqnarray*}
p_{\rm inv} & = &
\sum_{y\in f_n(W_n)}\Pr[y=f_n(U_n)] p_{\rm inv}(y)\\
& = & \sum_{y\in f_n(W_n)}
	\left\| \left(\left(\sum_{z\in f_n^{-1}(y)}{\cal P}^{z,\rm open}\right)
	\otimes {\cal P}_{\times}^{y,\rm commit}\right)
	{\cal U}_n {\cal P}_0 \ket{\tilde{\psi}_{n,0}}\right\|^2\\
& = & 	\left\| \left(\sum_{y\in f_n(W_n)}
\left(\left(\sum_{z\in f_n^{-1}(y)}{\cal P}^{z,\rm open}\right)
	\otimes {\cal P}_{\times}^{y,\rm commit}\right)\right)
	{\cal U}_n {\cal P}_0 \ket{\tilde{\psi}_{n,0}}\right\|^2\\
& = & \| {\cal P}_1 {\cal U}_n {\cal P}_0 \ket{\tilde{\psi}_{n,0}}\|^2.
\end{eqnarray*}

\noindent
Furthermore, we rewrite the above to easily estimate the value of $p_{\rm inv}$.
\begin{eqnarray*}
p_{\rm inv} & = &
\|{\cal P}_1 {\cal U}_n {\cal P}_0 \ket{\tilde{\psi}_{n,0}}\|^2
= \|{\cal P}_1 {\cal U}_n ({\cal I} - {\cal P}_0^{\bot})\ket{\tilde{\psi}_{n,0}}\|^2\\
& = & \|{\cal P}_1 {\cal U}_n \ket{\tilde{\psi}_{n,0}} - 
{\cal P}_1 {\cal U}_n {\cal P}_0^{\bot}\ket{\tilde{\psi}_{n,0}}\|^2\\
& = & \|{\cal P}_1 \ket{\tilde{\psi}_{n,1}} - 
	{\cal P}_1 {\cal U}_n {\cal P}_0^{\bot}\ket{\tilde{\psi}_{n,0}}\|^2.
\end{eqnarray*}
Using the triangle inequality and $b_1(n)>1-b_0(n)$, we have
\begin{eqnarray*}
p_{\rm inv} & \ge &
\left( \| {\cal P}_1\ket{\tilde{\psi}_{n,1}}\|
 - \| {\cal P}_1 {\cal U}_n {\cal P}_0^{\bot}\ket{\tilde{\psi}_{n,0}}\|\right)^2\\
& \ge & \left( \| {\cal P}_1\ket{\tilde{\psi}_{n,1}}\|
 - \| {\cal P}_0^{\bot}\ket{\tilde{\psi}_{n,0}}\|\right)^2\\
& = &  \left(\sqrt{b_1(n)} - \sqrt{1-b_0(n)}\right)^2. 
\end{eqnarray*}

Let us recall that we assume that $b_0(n)+b_1(n)>1+ 1/p(n)$ for some polynomial $p$.
After some calculation, we have 
\[ p_{\rm inv}\ge 2- 1/p - 2\sqrt{1-1/p} \ge 1/4(p(n))^2. \]
This 
%violates the assumption that $f$ is one-way, which 
completes the proof of Theorem \ref{lem:L-bind}.
%\hfill$\Box$ \vspace*{3mm}
\end{proof}

%\section{1-out-of-2 binding commitment from regular quantum one-way function with the unknown preimage size}
\section{1-out-of-2 binding commitment from quantum one-way function}
A 1-out-of-2 binding (we denote by ${2\choose 1}$-binding) commitment scheme 
consists of two commitment schemes where one of the two commitment schemes 
satisfies the binding property and the other does not have to satisfy
the binding property.
In \cite{HNORV09}, Haitner {\sl et al.\/} introduced a notion of 1-out-of-2 binding
commitment schemes and gave a construction of ${2\choose 1}$-binding commitment
schemes based on one-way function. We also
consider a quantum version of ${2\choose 1}$-binding commitment scheme and 
construct a ${2\choose 1}$-binding quantum commitment scheme.

We define a {\em 2-parallel quantum bit commitment scheme} $\Pi=(\Pi_1,\Pi_2)$, which is 
a parallel composition of two non-interactive quantum bit commitment
schemes $\Pi_1$ and $\Pi_2$. 
%$\Pi$ is a two-party protocol between Alice and Bob.
At the beginning of the protocol $\Pi$, Alice has two bits $w_1$ and $w_2$.
$\Pi$ consists of two phases, Commit Phase and Reveal Phase, as the standard
bit commitment schemes do. In Commit Phase, Alice (in $\Pi$) invokes
Commit Phase of $\Pi_1$ and sends $w_1$-commitment state (of $\Pi_1$) to Bob. 
Also she invokes Commit Phase of $\Pi_2$ and sends $w_2$-commitment state (of $\Pi_2$)
to Bob. We call the joint state of the 
$w_1$-commitment state (of $\Pi_1$) and the $w_2$-commitment state (of $\Pi_2$)
{\em$(w_1,w_2)$-commitment state\/} (of $\Pi$).
In Reveal phase, Alice sends decommitments both of $\Pi_1$ and $\Pi_2$.
Bob accepts if the both decommitments are valid.
%Before mentioning the requirements for 2-parallel quantum bit commitment,
%we provide our 2-parallel quantum bit commitment protocol (called {\sf Protocol 1}).
%While a sequential composition is discussed in \cite{HNORV09}, 
%our protocol runs {\sf Base Protocol} twice in parallel.

Next, we would like to define ``computational 1-out-of-2 binding''. 
In the classical case, it is defined in terms of transcripts. In the quantum case,
the definition based on transcripts is not easy to handle with. Fortunately,
our protocol below has a classical inner-state which controls the 1-out-of-2
binding property. Thus, after providing our protocol, we will give a
protocol-specific definition of computational 1-out-of-2 binding. 
Moreover, we discuss the hiding property later.

%\medskip
%\begin{definition}\rm
%Let $\ket{\psi_{w_1,w_2}}$ be a quantum state just after the end of the
%commit phase when Alice has inputs $w_1$ and $w_2$ and let
%$\rho_B(w_1,w_2)={\rm tr}_A(\ket{\psi_{w_1,w_2}}\bra{\psi_{w_1,w_2}})$.
%A 2-parallel quantum bit commitment is
%{\em strongly-hiding\/}
%if $\delta(\rho, \rho')$ is negligible
%for any distinct quantum states $\rho$ and $\rho'$ 
%in $\{\rho_B(w_1,w_2) : w_1,w_2\in\{0,1\}\}$.
%if both halves of the 2-parallel quantum bit commitment 
%\end{definition}

%\begin{definition}\rm
%A 2-parallel quantum bit commitment $\Pi$ is {\em strongly-hiding\/}
%if $\delta(\rho, \rho')$ is negligible for any distinct 
%$(w_1,w_2)$-commitment states $\rho$ and $\rho'$ for $\Pi$, where $(w_1,w_2)\in\{0,1\}^2$.
%in $\{\rho_B(w_1,w_2) : w_1,w_2\in\{0,1\}\}$.
%if both halves of the 2-parallel quantum bit commitment 
%\end{definition}

%Before mentioning the requirements for 2-parallel quantum bit commitment,
%Now, we are ready to provide
%The following is
We give
our 2-parallel quantum bit commitment protocol (called {\sf Protocol 1})
in Figure \ref{fig:p1}.
While a sequential composition is discussed in \cite{HNORV09}, 
our protocol runs {\sf Base Protocol} twice in parallel.

%\medskip
\begin{figure}[htbp]
\hrule
\medskip
\noindent
{\bf Parameters}: Integers $t\in [1,n]$, 
$\Delta_1\in [0,t]$ and $\Delta_2\in [0,n-t]$.\\[.5em]
{\bf Commit Phase}:
\begin{enumerate}\parsep=0pt\itemsep=0pt
\item Alice with her two bits $w_1$ and $w_2$
first chooses $x\in \{0,1\}^n$ uniformly and computes $y=f_n(x)$.
She also randomly chooses two hash functions $h_1$ and $h_2$ from families of
pairwise independent hash functions $H^{(1)}=\{ h_1: \{0,1\}^n\rightarrow
\{0,1\}^{t-\Delta_1}\}$ and $H^{(2)}=\{ h_2: \{0,1\}^n\rightarrow
\{0,1\}^{n-t-\Delta_2}\}$, respectively.
\item Next, Alice sends the quantum state 
\[ \ket{h_1,h_1(y)}_{\theta(w_1)} \otimes \ket{h_2,h_2(x)}_{\theta(w_2)}\in 
{\sf H}_{{\rm commit}_1}\otimes {\sf H}_{{\rm commit}_2} \]
to Bob.
\item 
Bob then stores the received quantum state $\rho_B$ until the reveal phase.
\end{enumerate}
{\bf Reveal Phase}:
\begin{enumerate}\parsep=0pt\itemsep=0pt
\item Alice announces the first decommitment $(w_1,h_1,y)$ 
and the second decommitment $(w_2,h_2,x)$ to Bob.
\item Next, Bob measures the first register of $\rho_B$ with measurement
$\{P^{h,z}_{\theta(w_1)}\}_{h\in H^{(1)},z\in {\it range}(h_1)}$ and 
obtains the classical output $(h,z)\in H^{(1)}\times {\it range}(h_1)$. 
Also he simultaneously measures the second register with measurement
$\{P^{h',z'}_{\theta(w_2)}\}_{h'\in H^{(2)},z'\in {\it range}(h_2)}$ and 
obtains the classical output $(h',z')\in H^{(2)}\times {\it range}(h_2)$. 
\item Lastly, Bob accepts the first commitment if and only if $h(y)=z$.
Also he accepts the second commitment if and only if $h'(z')=x$ and $y=f_n(x)$.
\end{enumerate}

\hrule
\caption{{\sf Protocol 1}}\label{fig:p1}
\end{figure}

\begin{definition}\rm
{\sf Protocol 1}
%A 2-parallel quantum bit commitment 
is
{\em computationally 1-out-of-2-binding\/}
if there exists a set $S\subseteq \{0,1\}^n$ such that
for every function $\varepsilon(n)=1/{\rm poly}(n)$,
the first half of the 2-parallel quantum bit commitment
is $\varepsilon(n)$-computationally-binding on condition that
a randomly chosen $x$ falls into $S$ and
the second half is $\varepsilon(n)$-statistically-binding
%satisfies that $\delta(\rho_B(0),\rho_B(1)) \ge 1-\varepsilon(n)$
on condition that $x$ does not fall into $S$.
\end{definition}

Next, we define the hiding property. Unfortunately, the hiding property
of {\sf Protocol 1} is not so strong. This is because the preimage-size
of $f$ is not constant over the inputs. Thus, we consider the following
weak definition of the binding property.

%The proof for Theorem \ref{thm:2para} does not use the regular property.
%It also goes through for Theorem \ref{thm:fowfbind}. On the other hand, the
%strong hiding does not hold any more. Fortunately, a hiding property
%in a weak sense holds. We consider the following hiding property.

\begin{definition}\rm
If, for any $\gamma$ with $0\le \gamma\le 1$, there exists a subset
$\Gamma \subseteq\{0,1\}^n$ satisfying the following
two properties, then  2-parallel quantum bit commitment is
$\gamma$-{\em hiding}.
\begin{enumerate}
\parsep=0pt\itemsep=0pt
\item $|\Gamma| \ge  \gamma \cdot 2^n$.
\item Let $w_1,w_2\in\{0,1\}$.
Let $Z_1(w_1)$ be a random variable for the first half
of the commitment when $w_1$ is the first bit to be committed
and $Z_2(w_2)$ be a random variable for the second half
when $w_2$ is the second bit to be committed, on condition that
$x$ is uniformly chosen from $\Gamma$.
Namely, the value for $(Z_1,Z_2)$ takes $\ket{h,h(f_n(x))}_{\theta(w_1)}\otimes
\ket{h',h'(x)}_{\theta(w_2)}$ where
$h$ is uniformly chosen from $H^{(1)}$, $h'$ is uniformly chosen from $H^{(2)}$ 
and $x$ is uniformly chosen from $\Gamma$.
Then, $(Z_1(0),Z_2(0))$, $(Z_1(0),Z_2(1))$, $(Z_1(1),Z_2(0))$ and
$(Z_1(1),Z_2(1))$ are negligibly close to each other.
\end{enumerate}
\end{definition}

\begin{theorem}\label{thm:2para}
Let $f=\{f_n:\{0,1\}^n \rightarrow \{0,1\}^{n}\}_{n\in\nat}$ be
an $s(n)$-secure regular quantum one-way function family, where $s(n)=n^{\omega (1)}$.
Then {\sf Protocol 1} with setting of parameters 
%$t=\lceil {\bf H}_2(f(U_n)) \rceil$ and 
$\Delta_1=\Delta_2=\frac{1}{4}\log s(n)$, is a 2-parallel quantum bit commitment scheme
that is 
%strongly-hiding and 
computationally 1-out-of-2 binding, regardless of the setting~of~$t$.
\end{theorem}

%Our proof strategy is similar to the proof in \cite{HNORV09}. 
While {\sf HNORV} scheme is two sequential of commitment schemes,
ours is two parallel of quantum commitment schemes. Thus, we
have to take into the account that the second half of the commitment might 
increase the power of the adversary. Fortunately, such information can be
included in the adversary's private space ${\sf H}_{\rm keep}$ and we
can use the same reduction as in {\sf Base Protocol}.
%discussed in the previous section. 
Actually, this observation
plays an important role through the paper.
%Besides the correlation between the first and the second commitment states,
Moreover, there is another difficulty in the analysis. Since the computational property
of ${2\choose 1}$-binding commitment is conditional (i.e., $x\in S$), we have
to consider the reduction between two algorithms
whose input distributions are different. To overcome the difficulty,
we use Non-interactive Quantum Hashing Theorem. 
While the proof of the computational part is similar to the proof in \cite{HNORV09},
the proof of the statistical part is completely different from the proof in \cite{HNORV09}
because it involves the analysis of quantum states.
%the new interactive hashing theorem is used in \cite{HNORV09}. 
%Instead of the new interactive hashing theorem, we use
%Theorem \ref{lem:L-bind} (Non-interactive Quantum Hashing Theorem). 
%Since Lemma \ref{lem:L-bind} is not a Our analysis 
%uses a combination of the quantum rewinding technique \cite{MW05,Watrous09}
%The proof of Theorem \ref{thm:2para} is given in Appendix \ref{sec:2para}.

\begin{proof}%(Theorem \ref{thm:2paraB})~
For every $t\in [1,n]$, we define the set of ``heavy'' strings to be
\[S_t = \{ x\in f_n^{-1}(y) :  \Pr[f_n(U_n)=y]\ge 2^{-t-\Delta_3}\} \] 
for the parameter $\Delta_3=\frac{1}{2}s(n)$.

We will show that if $x\in S_t$ is chosen in the first step of Commit Phase
then the first half is binding and if $x\not\in S_t$ then the second half is binding.

First, we show a reduction from inverting $f_n$ to violating the binding property
of {\sf Protocol 1} in the case of $x\in S_t$. 
Let $f_n':H^{(1)}\times \{0,1\}^n \rightarrow H^{(1)}\times \{0,1\}^{t-\Delta_1}$ 
be a function that maps $(h,x)$ to $(h,h(f_n(x)))$.
We define $R'_n = \{ (f'_n(h,x), (h,x) ) : x\in S_t~\mbox{and}~h\in H^{(1)}\}$
and $W_{h,\eta}=\{\udn{x} : (\eta, (h,x))\in R_n'\}$.

Let ${\cal A}_1$ be a quantum algorithm to violate the binding property
(with respect to $R_n'$) of {\sf Protocol 1}
%in the case of $x\not\in S_t$
with probability $\varepsilon(n)$.
Then, from Theorem \ref{lem:L-bind},
we have another algorithm ${\cal A}_2$ that inverts
$f_n'(h,x)$. Namely,
\[ \Pr[{\cal A}_2(H^{(1)},H^{(1)}(f_n(U_n)))\in W_{H^{(1)},H^{(1)}(f_n(U_n))}]
\ge \varepsilon(n)^2/4. \]
%Furthermore, we consider yet another algorithm ${\cal A}_3$ that
%boosts the success probability of ${\cal A}_2$.
For each $h\in H^{(1)}$ and $x\in \{0,1\}^n$, we consider
\[ p_{h,x} = 
\frac{\Pr[ {\cal A}_2(f_n'(h,x)) \in W_{h,h(f_n(x))}]}{|f_n'^{-1}(f_n'(h,x))|} \]
and set
\[ T = \{ (h,x) : p_{h,x}\ge \frac{\varepsilon(n)^2}{8}\}.\]
By the counting argument, we have $|T|\ge \varepsilon^2(n)/8\cdot 2^n|H^{(1)}|$.
Here, we estimate the following probability:
%$\Pr[{\cal A}_2(H^{(1)}, U_{t-\Delta_1})\in W_{H^{(1)},U_{t-\Delta_1}}]$. 
\begin{eqnarray*}
\lefteqn{\Pr[{\cal A}_2(H^{(1)}, U_{t-\Delta_1})\in W_{H^{(1)},U_{t-\Delta_1}}]}\\
& \ge & \sum_{(h,x)\in T}
\Pr[{\cal A}_2(h, h(f_n(x)))\in W_{h,h(f_n(x))}]
\cdot \Pr[H^{(1)}=h \land U_{t-\Delta_1}=h(f_n(x))]\\
&\ge& \varepsilon^4(n)/64.
\end{eqnarray*}

%By using the quantum rewinding technique \cite{MW05,Watrous09},
%it is not hard to show that for every $(h,x)\in T$, we can invert $f_n'(h,x)$ with
%probability $1 - 2^{-\Omega(n)}$. We call the amplification procedure ${\cal A}_3$.

We consider an algorithm $\cal B$ that on input $y=f_n(x)$, 
picks randomly a hash function $h\in H^{(1)}$,
and outputs ${\cal A}_2(h,h(y))$. 
%We define $W_{h,\eta}=\{\udn{x} : \eta=h(f_n(x))~\mbox{and}~x\not\in S_t\}$.
%and $W=\{(h,\eta): h\in H^{(1)}~\mbox{and}~\eta=h(f_n(x))~\mbox{for some}~x\not\in S_t\}$.
We analysis the probability that $\cal B$ inverts $f_n$ in the following.
%(Here, the security analysis. Namely, either of the following two holds.
%If the adversary breaks the first
%commitment, then an inverter is constructible. Or,
%the second commitment cannot open in two ways in information theoretical sense.)
\begin{eqnarray*}
\lefteqn{\Pr[{\cal B}(f_n(U_n))\in f_n^{-1}(f_n(U_n))]}\\
& = &  {\bf E}_{h\leftarrow {H^{(1)}}}[\Pr[{\cal A}_2(h,h(f_n(U_n)))\in f^{-1}(f_n(U_n))]]\\
& = &  {\bf E}_{h\leftarrow {H^{(1)}}}\left[\sum_{\udn{x}}
\Pr[f_n(U_n)=f_n(x) \land  {\cal A}_2(h,h(f(x_n)))=x]\right]\\
& = &  {\bf E}_{h\leftarrow {H^{(1)}}}\left[\sum_{\eta,x~\mbox{s.t.}~\eta=h(f_n(x))}
\Pr[f_n(U_n)=f_n(x)]\cdot\Pr[{\cal A}_2(h,\eta)=x]\right]\\
& \ge &  {\bf E}_{h\leftarrow {H^{(1)}}}\left[\sum_{\eta,x~\mbox{s.t.}~x\in W_{h,\eta}}
\Pr[f_n(U_n)=f_n(x)]\cdot\Pr[{\cal A}_2(h,\eta)=x]\right]\\
& \ge &  2^{-t-\Delta_3}\cdot
{\bf E}_{h\leftarrow {H^{(1)}}}\left[\sum_{\eta,x~\mbox{s.t.}~x\in W_{h,\eta}} \Pr[{\cal A}_2(h,\eta)=x]\right]\\
%& = &  2^{-t-\Delta_3}\cdot 2^{t-\Delta_1}\cdot \Pr_{h\leftarrow H^{(1)},\eta\leftarrow\{0,1\}^{t-\Delta_1}}[A(h,\eta)\in W_{h,\eta}]\\
& = &  2^{-t-\Delta_3} \cdot 2^{t-\Delta_1}\cdot \Pr[{\cal A}_2(H^{(1)},U_{t-\Delta_1})\in W_{H^{(1)},U_{t-\Delta_1}}]\\
& \ge & 2^{-(\Delta_1+\Delta_3)} \cdot \frac{\varepsilon(n)^4}{64}\\
& = & s(n)^{-3/4}\cdot \frac{\varepsilon(n)^4}{64},
%& \ge & 2^{n-(t-\Delta_1)} \cdot 2^{-(\Delta_1+\Delta_3)} \cdot \frac{\varepsilon(n)^2}{8} \cdot (1-2^{-\Omega(n)})\\
%& \ge & 2^{-(\Delta_1+\Delta_3)} \cdot \frac{\varepsilon(n)^2}{8} \cdot (1-2^{-\Omega(n)}).
\end{eqnarray*}
which is greater than $1/s(n)$ if $\varepsilon$ is non-negligible.

Next, we consider the case $x\not\in S_t$.
We define $W_y=\{(h,h(x)):h\in H^{(2)}~\mbox{and}~x\in f_n^{-1}(y)\}\subseteq \{0,1\}^q$,
where $q$ is the length of $(h,h(x))$.
Any (possibly cheating) quantum state $\ket{\psi}$ for the second commitment 
can be written as follows:
\[ \ket{\psi} = \sum_{z\in W_y} \alpha_z\ket{z}_{+} + \sum_{z\not\in W_y} \alpha_z\ket{z}_{+}, \]
since $\{\ket{z}_{+}\}_{z\in\{0,1\}^q}$ is a basis.
Then $b_0(n)=\sum_{z\in W_y} |\alpha_z|^2$,
Since $\ket{\psi}$ can be written as
\[ \ket{\psi} = \sum_{u\in W_y} \sum_{z\in \{0,1\}^q} \alpha_z (-1)^{\langle u,z\rangle}
\ket{u}_{\times} + 
\sum_{u\not \in W_y} \sum_{z\in \{0,1\}^q} \alpha_z (-1)^{\langle u,z\rangle} 
\ket{u}_{\times}, \]
\[b_1(n)=\frac{\displaystyle \sum_{u\in W_y}\left| \sum_{z\in \{0,1\}^q} \alpha_z (-1)^{\langle u,z\rangle}
\right|^2}{2^q}. \]
%\[b_1(n)=\frac{|\{x:f_n^{-1}(y)\}|}{2^{n-t-\Delta_2}}\le 
%	\frac{2^{n-t-\Delta_3}}{2^{n-t-\Delta_2}} = 2^{\Delta_2-\Delta_3} 
%	=\frac{1}{\sqrt[4]{s(n)}} \]
%regardless of $\alpha_z$. 
To maximize $b_0(n)+b_1(n)$, we set $a=\sum_{z\in W_y} |\alpha_z|^2=b_0(n)$.
On the condition that $b_0(n)=a$, $b_1(n)$ achieves the maximum when $|\alpha_z|$
is uniformly distributed. Actually, it is sufficient to consider the case where
\[ \alpha_z = \left\{ \begin{array}{@{}ll}\sqrt{a/|W_y|} & \mbox{\rm if}~z\in W_y\\
		\sqrt{(1-a)/(2^q-|W_y|)} & \mbox{otherwise.}
	\end{array}\right.\]

Then, we have $b_1(n)= (\sqrt{a|W_y|} + \sqrt{(1-a)(2^q-|W_y|)})^2/2^q$.
Let $\xi = |W_y|/2^q$. Thus, we have
\[ b(n)= 1 + (2a-1)\xi + 2\sqrt{a(1-a)\xi(1-\xi)}. \]
After some calculation, we have $b(n)\le 1 + \sqrt{\xi}$.
Since
\[ \xi \le \frac{|\{x:f_n^{-1}(y)\}|}{2^{n-t-\Delta_2}}\le 
	\frac{2^{n-t-\Delta_3}}{2^{n-t-\Delta_2}} = 2^{\Delta_2-\Delta_3} 
	=\frac{1}{\sqrt[4]{s(n)}}=\frac{1}{n^{\omega(1)}}, \]
we can say that $b(n)\le 1 + 1/n^{\omega(1)}$.
%\qed
\end{proof}
%\hfill$\Box$ \vspace*{3mm}

%Before showing Theorem \ref{thm:2para}, we prepare a technical lemma.
\begin{theorem}\label{thm:fowfhide}
Let $f=\{f_n:\{0,1\}^n \rightarrow \{0,1\}^{n}\}_{n\in\nat}$ be
an $s(n)$-secure quantum one-way function family, where $s(n)=n^{\omega(1)}$.
Then, there exists $t=t_0\in[1,n]$ such that
{\sf Protocol~1} satisfies $(1/n)$-hiding if we set
%$\Delta_1\ge \log n +4$ and $\Delta_2\ge 3$.
$\Delta_1=\Delta_2= \frac{1}{4}\log s(n)$.
\end{theorem}

First, we suppose that $f$ is a regular quantum one-way function.
Then the preimage-size is always constant. This means that ${\bf H}_2(f(U_n))$ is 
also constant. If the parameter $t$ is correctly given (i.e., 
$t={\bf H}_2(f(U_n))$), the first and the second commitment states are
almost maximally mixed. Though there is a small amount of correlation
between the first and the second commitment states, we can regard
the joint state as a quantum state close to the maximally mixed state
by the following lemma, which is a quantum version of Lemma 2.5 in \cite{GRS06}.

%Thus, the quantum state that Alice sends to Bob in Commit Phase is close to
%the maximally mixed state
%{\rm (e.g., \cite{GRS06})}
\begin{lemma}\label{lem:grs}
Let $\rho$ be a mixed state such that $\rho=\sum_x p_x\ket{x}\bra{x}\otimes \rho_x$.
If there exists a mixed state $\sigma'$ such that $\delta(\rho_x,\sigma')\le \varepsilon$
for all $x$, then $\delta(\rho,\sigma)\le \varepsilon$, where
$\sigma=\sum_x p_x\ket{x}\bra{x}\otimes \sigma'$.
\end{lemma}

\begin{proof}
\[ \delta(\rho,\sigma)=\frac{1}{2}{\rm tr}\sqrt{(\rho-\sigma)^\dag(\rho-\sigma)}
=\frac{1}{2}\sum_x p_x {\rm tr}\sqrt{(\rho_x-\sigma')^\dag(\rho_x-\sigma')}
\le \varepsilon \sum_x p_x = \varepsilon. \]%\qed
\end{proof}

\begin{proof}~(Theorem \ref{fowfhide})
We will see the first property.
%\subsection{1st Property}
Let $p(y)=\Pr[f_n(U_n)=y]$ and $\mu(X)=|X|/2^n$.
For any $t\in [1,n]$, let $A_t=\{y\in \{0,1\}^n :
2^{-t}\le p(y) < 2^{-t+1}\}$. Since
$\bigcup_t A_t = f_n(\{0,1\}^n)$, there exists $t_0$
such that $\Pr[f_n(U_n)\in A_{t_0}]\ge 1/n$. Then,
we define $\Gamma_1$ and $\Gamma_2$ as follows:
\begin{eqnarray*}
\Gamma_1 & = & \{ \udn{x}: p(f_n(x))<2^{-t_0+1}\}~\mbox{and}\\
\Gamma_2 & = & \{ \udn{x}: p(f_n(x))\ge 2^{-t_0}\},
\end{eqnarray*}
and set $\Gamma=\Gamma_1\cap \Gamma_2$.
Thus, it is easy to see that $\Gamma_1\cup \Gamma_2 = \{0,1\}^n$
and $\mu(\Gamma) \ge 1/n$.%\qed

\medskip

Next, we will see the second property.
%\subsection{2nd Property}
%The second property can be shown by a discussion similar to the proof of
%Theorem \ref{thm:2paraH}.
Let $C_1$ and $C_2$ be subsystems for ${\sf H}_{{\rm commit}_1}$ and
${\sf H}_{{\rm commit}_2}$, respectively. 
Assume that Alice has two bits $w_1$ (resp., $w_2$) for the first (resp., the second) 
half of the commitment. Also assume that $x$ falls into $\Gamma$.

Let $\rho$ be the quantum state ${\rm tr}_{C_2}(\rho_B)$. Then, $\rho$ 
can be written as 
\begin{eqnarray*}
\rho & = & \sum_{x\in {\Gamma},h\in H^{(1)}} 
\frac{1}{|\Gamma|\cdot |H^{(1)}|}\ket{h,h(f_n(x))}_{\theta(w_1)}\bra{h,h(f_n(x))}.
\end{eqnarray*}
Let
\begin{eqnarray*}
\iota_{+} & = & \sum_{z\in\{0,1\}^{t-\Delta_1},h\in H^{(1)}}
	\frac{1}{2^{t-\Delta_1}|H^{(1)}|} \ket{h,z}_{+}\bra{h,z} \quad\mbox{and}\\
\iota_{\times} & = & \sum_{z\in\{0,1\}^{t-\Delta_1},h\in H^{(1)}}
	\frac{1}{2^{t-\Delta_1}|H^{(1)}|} \ket{h,z}_{\times}\bra{h,z}.
\end{eqnarray*}
Then $\iota\stackrel{\rm def}{=}\iota_{+}=\iota_{\times}$ is the uniform distribution.
By Leftover Hash Lemma, we have $\delta(\rho,\iota)\le 2^{-\Delta_1/2}$.

Next we let $\rho'(y)$ 
be the quantum state ${\rm tr}_{C_1}(\rho_B)$ when $y=f_n(x)$ is given.
(Note that any elements in $f_n^{-1}(y)$ are also in $\Gamma$.
Thus, the likelihood of $x$ is the same as that of any other $x'\in f_n^{-1}(y)$.)
Then, $\rho'(y)$ can be written as 
\begin{eqnarray*}
\rho'(y) & = & \sum_{x\in f_n^{-1}(y),h\in H^{(2)}} \frac{1}{|f_n^{-1}(y)|\cdot|H^{(2)}|}\ket{h,h(x)}_{\theta(w_2)}\bra{h,h(x)}.
\end{eqnarray*}
Let
\begin{eqnarray*}
\iota_{+}' & = & \sum_{z\in\{0,1\}^{n-t-\Delta_2},h\in H^{(2)}}
	\frac{1}{2^{n-t-\Delta_2}|H^{(2)}|} \ket{h,z}_{+}\bra{h,z} \quad\mbox{and}\\
\iota_{\times}' & = & \sum_{z\in\{0,1\}^{n-t-\Delta_2},h\in H^{(2)}}
	\frac{1}{2^{n-t-\Delta_2}|H^{(2)}|} \ket{h,z}_{\times}\bra{h,z}.
\end{eqnarray*}
Then $\iota'\stackrel{\rm def}{=}\iota_{+}'=\iota_{\times}'$ is the uniform distribution.
By Leftover Hash Lemma, we have $\delta(\rho'(y),\iota')\le 2^{-\Delta_2/2}$ for any $y$.
By Lemma \ref{lem:grs} and the triangle inequality,
we have 
\[ \delta(\rho_B(w_1,w_2), (\iota,\iota')) \le 2^{-\Delta_1/2} + 2^{-\Delta_2/2} =
\frac{2}{\sqrt[8]{s(n)}} = \frac{1}{n^{\omega(1)}} \]
for any $w_1,w_2\in\{0,1\}$.%\qed
\end{proof}

\section{Hiding Amplification}
In the previous section, we showed that {\sf Protocol 1} based on
quantum one-way function holds a ``weak'' hiding property.
In this section, we amplify the ``weak'' hiding property to a ``strong'' one. 
We consider $n^2$ parallel executions
of {\sf Protocol 1} to amplify the hiding probability
from $1/n$ to $1-2^{-n^{\Omega(1)}}$.
We describe the resulting parallel protocol (called {\sf Protocol 2}) in Figure \ref{fig:p2}.
%in the following.

%\medskip
\begin{figure}[htbp]
\hrule
\medskip
\noindent
{\bf Parameters}: Integers $t\in [1,n]$, 
$\Delta_1\in [0,t]$ and $\Delta_2\in [0,n-t]$.\\[.5em]
{\bf Commit Phase}:
\begin{enumerate}\parsep=0pt\itemsep=0pt
\item Alice with her two bits $w_1$ and $w_2$
first chooses $x_1,\ldots, x_{n^2}\in (\{0,1\}^n)^{n^2}$ uniformly and computes 
$y_1=f_n(x_1),\ldots, y_{n^2}=f_n(x_{n^2})$.
Also, she uniformly and independently chooses pairwise independent hash functions 
$h_{1,1},\ldots,h_{1,n^2}\in (H^{(1)})^{n^2}$ and 
$h_{2,1},\ldots,h_{2,n^2}\in (H^{(2)})^{n^2}$.
\item Alice chooses $w_{1,1},\ldots, w_{1,n^2}\in (\{0,1\})^{n^2}$
and $w_{2,1},\ldots, w_{2,n^2}\in (\{0,1\})^{n^2}$ such that
$w_1=w_{1,1}\oplus\cdots \oplus w_{1,n^2}$ and $w_2=w_{2,1}\oplus\cdots \oplus w_{2,n^2}$.
\item Next, Alice sends the quantum state 
\[ \begin{array}{l}
\ket{h_{1,1},h_{1,1}(y_1)}_{\theta(w_{1,1})}\otimes \cdots \otimes
		\ket{h_{1,n^2},h_{1,n^2}(y_{n^2})}_{\theta(w_{1,n^2})}\\
\hspace*{5mm}
 \otimes \ket{h_{2,1},h_{2,1}(x_1)}_{\theta(w_{2,1})}\otimes \ldots \otimes
		\ket{h_{2,n^2},h_{2,n^2}(x_{n^2})}_{\theta(w_{2,n^2})}
\in 
({\sf H}_{{\rm commit}_1})^{\otimes n^2}\otimes ({\sf H}_{{\rm commit}_2})^{\otimes n^2}
\end{array}\]
to Bob.
\item 
Bob then stores the received quantum state $\rho_B$ until the reveal phase.
\end{enumerate}
{\bf Reveal Phase}:
\begin{enumerate}\parsep=0pt\itemsep=0pt
\item Alice announces the first decommitments $(w_{1,1},h_{1,1},y_1),\ldots,
(w_{1,n^2},h_{1,n^2},y_{n^2})$ 
and the second decommitments $(w_{2,1},h_{2,1},x_1),\ldots,(w_{2,n^2},h_{2,n^2},x_{n^2})$
to Bob.
\item Next, Bob measures the first register of $\rho_B$ with measurement
\[ 
\{P^{h,z}_{\theta(w_{1,1})}\}_{h\in H^{(1)},z\in {\it range}(h_{1,1})} 
\otimes \cdots \otimes
\{P^{h,z}_{\theta(w_{1,n^2})}\}_{h\in H^{(1)},z\in {\it range}(h_{1,n^2})}
\]
and 
obtains the classical output $(h_1,z_1,\ldots,h_{n^2},z_{n^2})$, where
each $(h_i,z_i)$ is in $H^{(1)}\times {\it range}(h_{1,i})$.
Also he simultaneously measures the second register with measurement
\[
\{P^{h',z'}_{\theta(w_{2,1})}\}_{h'\in H^{(2)},z'\in {\it range}(h_{2,1})}
\otimes \cdots \otimes
\{P^{h',z'}_{\theta(w_{2,n^2})}\}_{h'\in H^{(2)},z'\in {\it range}(h_{2,n^2})}
\]
and 
obtains the classical output $(h_1',z_1',\ldots, h_{n^2}',z_{n^2}')$, where
each $(h_i',z_i')$ is in $H^{(2)}\times {\it range}(h_{2,i})$.
\item Lastly, Bob accepts the first commitment if and only if $h_i(y_i)=z_i$
for every $i$, and recovers the first committed bit $w_1$ as
$w_1=w_{1,1}\oplus\cdots \oplus w_{1,n^2}$.
Also he accepts the second commitment if and only if $h_i'(z_i')=x_i$ and $y_i=f_n(x_i)$
for every $i$, and recovers the second committed bit $w_2$ as
$w_2=w_{2,1}\oplus\cdots \oplus w_{2,n^2}$.
\end{enumerate}

\hrule
\caption{{\sf Protocol 2}}\label{fig:p2}
\end{figure}

%\medskip

In the rest of this section, we state that
{\sf Protocol 2} achieves $(1-2^{-n^{\Omega(1)}})$-hiding
and preserves the binding property.

\begin{theorem}\label{thm:stronghide}
Let $f=\{f_n:\{0,1\}^n \rightarrow \{0,1\}^{n}\}_{n\in\nat}$ be
an $s(n)$-secure quantum one-way function family, where $s(n)=n^{\omega (1)}$.
Then, there exists $t=t_0\in[1,n]$ such that
{\sf Protocol~2} satisfies $(1-2^{-n^{\Omega(1)}})$-hiding if we set
%$\Delta_1\ge \log n +4$ and $\Delta_2\ge 3$.
$\Delta_1=\Delta_2= \frac{1}{4}\log s(n)$.
\end{theorem}

Due to the parallel composition, the proof becomes quite simpler than
the proof of the hiding amplification in \cite{HNORV09}. The proof of Theorem
\ref{thm:stronghide} can be done by a standard probabilistic argument.
% (in Appendix \ref{sec:stronghide}).

\begin{proof}
We will see the first property.
%\subsection{1st Property}
Recall that $\Gamma= \{x\in\{0,1\}^n: 2^{-t_0}\le p(f_n(x)) < 2^{-t_0+1} \}$ for some $t_0$.
Let $N=n^3$ and $\Gamma'=\{(x_1,\ldots,x_{n^2})\in (\{0,1\}^n)^{n^2}:
\exists i, x_i\in \Gamma\}$. 
We consider the probability $p$ where some $x_i$ falls in $\Gamma$.
Since $\mu(\Gamma)\ge 1/n$, we have that $p>1- (1-1/n)^{n^2}$.
By the fact that $1-t < e^{-t}$, we have $p\ge 1-e^{-n}>1-2^{-n}=1-2^{-N^{1/3}}$.%\qed

\medskip

Next, we will see the second property.
%\subsection{2nd Property}
Let $(x_1,\ldots, x_{n^2})\in\Gamma'$. By the definition of $\Gamma'$, we may
assume that $x_J\in\Gamma$ for some $J\in [1,n^2]$.
Recall that
$Z_1(w_1)$ is of the form 
\[\ket{h_{1,1},h_{1,1}(f_n(x_1))}_{\theta(w_{1,1})}\otimes \cdots\otimes
\ket{h_{1,n^2},h_{1,n^2}(f_n(x_{n^2}))}_{\theta(w_{1,n^2})} \]
and
$Z_2(w_2)$ is of the form 
\[\ket{h_{2,1},h_{2,1}(x_1)}_{\theta(w_{2,1})}\otimes\cdots\otimes
\ket{h_{2,n^2},h_{2,n^2}(x_n^2)}_{\theta(w_{2,n^2})}.\]
Let $W(x_i,w_{1,i},w_{2,i})$ be the composition of the
$i$-th components of $Z_1(w_1)$ and $Z_2(w_2)$, that is,
\[W(x_i,w_{1,i},w_{2,i})=\ket{h_{1,i},h_{1,i}(f_n(x_i))}_{\theta(w_{1,i})}
\ket{h_{2,i},h_{2,i}(x_i)}_{\theta(w_{2,i})}.\]
Since $w_{1,1},\ldots,w_{1,n^2}$ and
$w_{2,1},\ldots,w_{2,n^2}$ are randomly chosen so as to satisfy that
$w_1=w_{1,1}\oplus \cdots \oplus w_{1,n^2}$ and 
$w_2=w_{2,1}\oplus \cdots \oplus w_{2,n^2}$, we may assume that
$w_{1,j}$ and $w_{2,j}$ with $j\ne J$ are uniformly and independently 
chosen from $\{0,1\}$ and $w_{1,J}$ and $w_{2,J}$ are determined by $w_1$, $w_2$, and
all $w_{1,j}$ and $w_{2,j}$ such that $j\ne J$. Thus, we can say that
$W(x_i,w_{1,i},w_{2,i})$ such that $i\ne J$ does not depend on the value
of $w_1$ and $w_2$.

On the other hand, $W(x_J,w_{1,J},w_{2,J})$ depends on the value $w_1$ and $w_2$. But,
we can show that it is $1/n^{\omega(1)}$-close to the uniform distribution 
by using the proof for the 2nd property of Theorem \ref{thm:fowfhide}.
From Lemma \ref{lem:grs}, we can say that $(Z_1(0),Z_2(0))$, $(Z_1(0),Z_2(1))$,
$(Z_1(1),Z_2(0))$ and $(Z_1(1),Z_2(1))$ are
$1/n^{\omega(1)}$-close to each other.
%\hfill$\Box$ \vspace*{3mm}
\end{proof}

\begin{theorem}\label{thm:parabind}
Let $f=\{f_n:\{0,1\}^n \rightarrow \{0,1\}^{n}\}_{n\in\nat}$ be
an $s(n)$-secure quantum one-way function family, where $s(n)=n^{\omega (1)}$.
Then {\sf Protocol 2} with setting of parameters 
$\Delta_1=\Delta_2=\frac{1}{4}\log s(n)$, is a 2-parallel quantum bit commitment scheme
that is computationally 1-out-of-2 binding regardless of the setting~of~$t$.
\end{theorem}

The above theorem says that {\sf Protocol 2} has 1-out-of-2 binding property.
Specifically speaking, the computational binding of the first half commitment
can be guaranteed in some case and the statistical binding of the second
half commitment can be guaranteed in the other case. The computational binding
of the first half commitment in {\sf Protocol 2} is reduced to the computational
binding of the first half commitment in {\sf Protocol 1}. The statistical
binding of the second half commitment in {\sf Protocol 2} can be shown by a
probabilistic argument.

To prove that {\sf Protocol 2} is computationally 1-out-of-2 binding, we have to
specify a set that controls the 1-out-of-2 property as in the proof of Theorem
\ref{thm:2para}. In the proof of Theorem \ref{thm:2para}, $S_t$ is such a set.
For the proof of Theorem \ref{thm:parabind}, we will use $S_t' 
= \{ (x_1,\ldots,x_{n^2})\in (\{0,1\}^n)^{n^2}: 
\exists i, x_i\in S_t\}$. Even if we can use the reduction to the $j$-th subprotocol,
we cannot know whether $x_j\in S_t$ or not. If $x_j\in S_t$, then the reduction
goes through. If we cannot assume that that $x_j\in S_t$, we can say that 
$x_j\in \{0,1\}^n$. Since we do not have to know some underlying relation to apply
Non-interactive Quantum Hashing Theorem, we can show that the reduction still goes
through. 
%The proof of Theorem \ref{thm:parabind} is given in Appendix \ref{sec:parabind}.

\begin{proof}
The proof is similar to the proof for Theorem \ref{thm:2para}.
For every $t\in [1,n]$, we define the set of ``heavy'' strings to be
\[S_t' = \{ (x_1,\ldots,x_{n^2})\in (\{0,1\}^n)^{n^2}: 
\exists i, x_i\in S_t\}, \]
%for the parameter $\Delta_3=\frac{1}{4}s(n)$,
where $S_t$ is defined in the proof of Theorem \ref{thm:2para}.

We will show that if $(x_1,\ldots,x_{n^2})\in S_t'$ is chosen in the 
first step of Commit Phase then the first half is binding and 
if $(x_1,\ldots,x_{n^2})\not\in S_t'$ then the second half is binding.

First, we show a reduction from inverting $f_n$ to violating the binding property
of {\sf Protocol 2} in the case of $(x_1,\ldots,x_{n^2})\in S_t'$. 
Recall that in the proof of Theorem \ref{thm:2para}
$f_n':H^{(1)}\times \{0,1\}^n \rightarrow H^{(1)}\times \{0,1\}^{t-\Delta_1}$ 
is a function that maps $(h,x)$ to $(h,h(f_n(x)))$,
$R'_n = \{ (f'_n(h,x), (h,x) ) : x\in S_t~\mbox{and}~h\in H^{(1)}\}$ and 
$W_{h,\eta}=\{\udn{x} : (\eta, (h,x))\in R_n'\}$.
We also define $R_n = \{ (f'_n(h,x), (h,x) ) : x\in \{0,1\}^n~\mbox{and}~h\in H^{(1)}\}$.

Let ${\cal A}_3$ be a quantum algorithm to violate the binding property
of {\sf Protocol 2} with probability $\varepsilon(n)$.
This means that ${\cal A}_3$ can send a quantum state in Commit Phase 
so that Bob can accept it either as $0$-commitment with probability $b_0(n)$ and 
as $1$-commitment with probability $b_1(n)$, where 
$b(n)=b_0(n)+b_1(n)\ge 1 + \varepsilon(n)$.
To make Bob accept the quantum state as a valid commitment in {\sf Protocol 2}, 
${\cal A}_3$ has to make Bob accept all executions of sub-protocol {\sf Protocol 1}.
Let $b_w^{(i)}(n)$ be the probability that ${\cal A}_3$ can make Bob accept
the $i$-th sub-protocol as $w$-commitment. We set $b^{(i)}(n)=b_0^{(i)}(n)+
b_1^{(i)}(n)$. 
Let $\tilde{b}_0(n)$ (resp., $\tilde{b}_1(n)$)
be the probability where ${\cal A}_3$ fails to make Bob accept
the quantum state as $0$-commitment (resp., $1$-commitment).
Similarly, we define $\tilde{b}_0^{(i)}(n)$ and $\tilde{b}_1^{(i)}(n)$ 
for each $i\in [1,n^2]$. Then, we have $\tilde{b}_0(n)+\tilde{b}_1(n)\le 1-\varepsilon$.
Since the failure probabilities are accumulative, there exists an index $j\in [1,n^2]$
such that $\tilde{b}_0^{(j)}(n)+\tilde{b}_1^{(j)}(n)\le 1-\varepsilon$.
Hence, we have $b_0^{(j)}(n)+b_1^{(j)}(n)\ge 1+\varepsilon$.
Thus, we can assume that a quantum algorithm ${\cal A}_4$
to violate the binding property of {\sf Protocol 1}
with probability $\varepsilon$.
Note that the violation (by ${\cal A}_4$) against the binding property of 
{\sf Protocol 1} is respect to either $R_n$ or $R_n'$. Fortunately, we do not have
to know which relation should be considered, since the algorithm with respect to
$R_n$ is the same as the one with respect to $R_n'$.
If ${\cal A}_4$ violates the binding property of the $j$-th sub-protocol
and $x_j\in S_t$, ${\cal A}_4$ does with respect to $R_n'$.
If ${\cal A}_4$ violates the binding property of the $j$-th sub-protocol
and $x_j\in \{0,1\}^n$, ${\cal A}_4$ does with respect to $R_n$.
Here, we consider only the case that ${\cal A}_4$ does with respect to $R_n'$,
since the other case is similar and easier to show.

From Theorem \ref{lem:L-bind}, we have another algorithm ${\cal A}_5$
satisfying that
\[ \begin{array}{l}\Pr[{\cal A}_5( H^{(1,1)},H^{(1,1)}(f_n(U_n^{(1)})),\ldots,
H^{(1,n^2)},H^{(1,n^2)}(f_n(U_n^{(n^2)})))=(j,z)\\
\hspace*{4cm}\land ~ z \in W_{H^{(1,j)},H^{(1,j)}(f_n(U_n^{(j)}))}]
\ge \varepsilon(n)^2/4,
\end{array}\]
where $H^{(1,1)},\ldots,H^{(1,n^2)}$ are independent and identical distributions
to $H^{(1)}$ and
$U_n^{(1)},\ldots,U_n^{(n^2)}$ are independent and identical distributions to $U_n$.

By the similar discussion in the proof of Theorem \ref{thm:2para},
we can say that
\[ \begin{array}{l}
\Pr[{\cal A}_5( H^{(1,1)},H^{(1,1)}(f_n(U_n^{(1)})),\ldots,
H^{(1,j-1)},H^{(1,j-1)}(f_n(U_n^{(j-1)})),
H^{(1,j)},U_{t-\Delta},\\
\hspace*{2cm}H^{(1,j+1)},H^{(1,j+1)}(f_n(U_n^{(j+1)})),\ldots,
H^{(1,n^2)},H^{(1,n^2)}(f_n(U_n^{(n^2)})))=(j,z)\\
\hspace*{4cm}\land ~ z \in W_{H^{(1,j)},U_{t-\Delta_1}}]
\ge \varepsilon(n)^4/64.
\end{array}\]

We consider an algorithm $\cal B$ that on input $y=f_n(x)$, 
picks randomly an integer $j'\in[1,n^2]$, a hash function $h\in H^{(1)}$.
$\cal B$ also picks randomly $x_1,\ldots,x_{j'-1},x_{j'+1},\ldots,x_{n^2}$
and $h_1,\ldots,h_{j'-1},h_{j'+1},\ldots,h_{n^2}$, computes
$y_1=f_n(x_1),\ldots, y_{j'-1}=f_n(x_{j'-1}),
y_{j'+1}=f_n(x_{j'+1}),\ldots, y_{n^2}=f_n(x_{n^2})$
and outputs the second part of 
${\cal A}_5(h_1,h_1(y_1),\ldots,h_{j'-1},h_{j'-1}(y_{j'-1}),\break
h,h(y),
h_{j'+1},h_{j'+1}(y_{j'+1}),\ldots,h_{n^2},h_{n^2}(y_{n^2}))$. 
Then, we have the following.
\begin{eqnarray*}
\lefteqn{\Pr[{\cal B}(f_n(U_n))\in f_n^{-1}(f_n(U_n))]}\\
& \ge &  {\bf E}_{h\leftarrow {H^{(1)}}}[\Pr[{\cal A}_5(h_1,h_1(f_n(U_n^{(1)})),\ldots,
h_{j'-1},h_{j'-1}(f_n(U_n^{(j'-1)})), h, h(f_n(U_n)),\\
&& \hspace*{2cm}h_{j'+1},h_{j'+1}(f_n(U_n^{(j'+1)})),\ldots, h_{n^2},h_{n^2}(f_n(U_n^{(n^2)})))=(j,z)\\
&& \hspace*{4cm}\land j=j' \land z \in f_n^{-1}(f_n(U_n))]]\\
& = & \frac{1}{n^2}\cdot
{\bf E}_{h\leftarrow {H^{(1)}}}[\Pr[{\cal A}_5(h_1,h_1(f_n(U_n^{(1)})),\ldots,
h_{j'-1},h_{j'-1}(f_n(U_n^{(j'-1)})), h, h(f_n(U_n)),\\
&& \hspace*{2cm} h_{j'+1},h_{j'+1}(f_n(U_n^{(j'+1)})),\ldots, h_{n^2},h_{n^2}(f_n(U_n^{(n^2)})))=(j,z)\\
&& \hspace*{4cm} \land z \in f_n^{-1}(f_n(U_n))]].
\end{eqnarray*}
The rest of the probabilistic analysis is similar to the proof of Theorem \ref{thm:2para}.
This shows that
\[ \Pr[{\cal B}(f_n(U_n))\in f_n^{-1}(f_n(U_n))]\ge s(n)^{-3/4}\cdot
\varepsilon(n)^4/64n^{2},  \]
which is greater than $1/s(n)$ if $\varepsilon$ is non-negligible.

Next, we consider the case $(x_1,\ldots, x_{n^2})\not\in S_t'$.
By the definition of $S_t'$, we have $x_i\not\in S_t$ for all $i$.
Thus, we can use the discussion for the proof of Theorem \ref{thm:2para}
for each bit $w_{2,i}$. This means that the value of $b^{(i)}(n)$ is less than
$1 + 1/n^{\omega(1)}$ for each bit $w_{2,i}$. 
Let us consider the event that Bob accepts the quantum state sent by Alice
in Commit Phase of {\sf Protocol 2} as $0$-commitment (or, $1$-commitment). 
Let $p$ be the probability
that this event occurs. Since this even occurs if Bob
accepts all decommitments of the sub-protocol, we can write $p=p_1\cdot\cdots p_{n^2}$,
where $p_i$ is either the probability that Bob accepts the quantum state
sent by Alice in Commit Phase of the $i$-th sub-protocol as $0$-commitment
or the probability that Bob accepts the quantum state
sent by Alice in Commit Phase of the $i$-th sub-protocol as $1$-commitment.
Thus, the best strategy for cheating is to behave honestly
for $n^2-1$ executions of the sub-protocol and maliciously for
just one execution. Hence, we can upper-bound $b(n)$ by 
$1^{n^2-1}(1+ 1/n^{\omega(1)})=1+ 1/n^{\omega(1)}$.%\qed
\end{proof}
%\hfill$\Box$ \vspace*{3mm}

\section{Statistically-Hiding Commitment from ${2\choose 1}$-Binding Commitment}
We have obtained the strongly-hiding 1-out-of-2 binding quantum commitment
based on quantum one-way function. But it is not a single scheme but a family
of scheme candidates. First, we construct a family of candidates for normal 
statistically-hiding quantum bit commitment from the family of candidates
for the strongly-hiding 1-out-of-2 binding quantum commitment. Next, we construct 
a single normal statistically-hiding quantum bit commitment from 
the family of candidates for the statistically-hiding quantum bit commitment.

\subsection{Statistically-Hiding Quantum Commitment Family from
${2\choose 1}$-Binding Quantum Commitment Family}

{\sf Protocol 2} consists of the first half commitment and the second half commitment.
We denote by ${\sf P2first}(w_1)$ the first half commitment with the committed bit $w_1$
and by ${\sf P2second}(w_2)$ the second half commitment with the committed bit $w_2$.
We consider the 
%following 
protocol (called {\sf Protocol 3}) in Figure \ref{fig:p3}.

%We denote by $Commit(i,t,w)$ the $i$-th half of the Commit Phase of {\sf Protocol 2} 
%with parameter $t$ and by $Reveal(i,t,w)$ the $i$-th half of the Reveal Phase.

\begin{figure}[htbp]
%\medskip
\hrule
\medskip
\noindent
{\bf Parameters}: Integers $t\in [1,n]$, 
$\Delta_1\in [0,t]$ and $\Delta_2\in [0,n-t]$. (These are succeeded to the sub-protocol
${\sf P2first}$ and ${\sf P2second}$.)\\[.5em]
{\bf Commit Phase}:
\begin{enumerate}\parsep=0pt\itemsep=0pt
\item Alice with her bit $w$
executes ${\sf P2first}(w)$ and ${\sf P2second}(w)$ in parallel.
%chooses $x\in \{0,1\}^n$ uniformly and computes $y=f_n(x)$.
%She also randomly chooses two hash functions $h_1$ and $h_2$ from families of
%pairwise independent hash functions $H^{(1)}=\{ h_1: \{0,1\}^n\rightarrow
%\{0,1\}^{t-\Delta_1}\}$ and $H^{(2)}=\{ h_2: \{0,1\}^n\rightarrow
%\{0,1\}^{n-t-\Delta_2}\}$, respectively.
%\item Next, Alice sends the quantum state 
%\[ \ket{h_1,h_1(y)}_{\theta(w)} \otimes \ket{h_2,h_2(x)}_{\theta(w)}\in 
%{\sf H}_{{\rm commit}_1}\otimes {\sf H}_{{\rm commit}_2} \]
%to Bob.
%\item 
%Bob then stores the received quantum state $\rho_B$ until the reveal phase.
\end{enumerate}
{\bf Reveal Phase}:
\begin{enumerate}\parsep=0pt\itemsep=0pt
\item Alice sends decommitments for ${\sf P2first}(w)$ and ${\sf P2second}(w)$
and Bob recovers the committed bits $w'$ and $w''$, respectively.
\item Bob verifies the correctness of the decommitments.
If the verification procedures for both 
${\sf P2first}(w)$ and ${\sf P2second}(w)$ are passed and $w'=w''$ 
%is Bob gets $w'$ and $w''$. If 
then Bob accepts.
%Next, Bob measures the first register of $\rho_B$ with measurement
%$\{P^{h,z}_{\theta(w)}\}_{h\in H^{(1)},z\in {\it range}(h_1)}$ and 
%obtains the classical output $(h,z)\in H^{(1)}\times {\it range}(h_1)$. 
%Also he simultaneously measures the second register with measurement
%$\{P^{h',z'}_{\theta(w)}\}_{h'\in H^{(2)},z'\in {\it range}(h_2)}$ and 
%obtains the classical output $(h',z')\in H^{(2)}\times {\it range}(h_2)$. 
%\item Lastly, Bob accepts the first commitment if and only if $h(y)=z$,
%$h'(z')=x$ and $y=f_n(x)$.
\end{enumerate}

\hrule
\caption{{\sf Protocol 3}}\label{fig:p3}
\end{figure}
%\medskip

\begin{theorem}\label{thm:nmfam}
Let $f=\{f_n:\{0,1\}^n \rightarrow \{0,1\}^{n}\}_{n\in\nat}$ be
an $s(n)$-secure quantum one-way function family, where $s(n)=n^{\omega (1)}$.
Then {\sf Protocol 3} with setting of parameters 
$\Delta_1=\Delta_2=\frac{1}{4}\log s(n)$, is a computationally-binding
quantum bit commitment scheme regardless of the setting~of~$t$.
Also, there exists $t=t_0\in [1,n]$ such that
{\sf Protocol 3} with the same parameter for $\Delta_1$ and $\Delta_2$
is statistically-hiding.
\end{theorem}

The hiding property can be shown by the argument in the proof of 
Theorem \ref{thm:stronghide}. Basically, {\sf Protocol 3} has the 1-out-of-2
binding property. Thus, either {\sf P2first} or {\sf P2second} must have
the binding property. Even if the adversary can violate either {\sf P2first} or 
{\sf P2second}, such violation can be detected by the equality check $w'=w''$ in
Reveal Phase.
Theorem \ref{thm:nmfam} can be similarly shown as 
Theorems \ref{thm:stronghide} and \ref{thm:parabind}.

\subsection{From a family To single BC}
As mentioned in Theorem \ref{thm:nmfam}, there exists a value $t$ such
that {\sf Protocol 3} has both the computational binding and statistical hiding.
But, we do not know the right value of $t$. By using a similar technique in
Section 7, we consider a combined protocol of {\sf Protocol} with different parameters.

${\sf P3}(t,w)$ denotes the commit phase of {\sf Protocol 3} 
with the committed value $w$ and parameter $t$. 
We consider the 
%following 
protocol (called {\sf Protocol 4}) in Figure \ref{fig:p4}.

%\medskip
\begin{figure}[htbp]
\hrule
\medskip
\noindent
{\bf Parameters}: Integers $\Delta_1\in [0,t]$ and $\Delta_2\in [0,n-t]$. 
(These are succeeded to the sub-protocol ${\sf P3}$.)\\[.5em]
{\bf Commit Phase}:
\begin{enumerate}\parsep=0pt\itemsep=0pt
\item Alice with her bit $w$ chooses $w_1,\ldots,w_n\in (\{0,1\})^n$
such that $w=w_1\oplus \cdots \oplus w_n$.
\item Alice executes ${\sf P3}(1,w_1),\ldots, {\sf P3}(n,w_n)$ in parallel.
\end{enumerate}
{\bf Reveal Phase}:
\begin{enumerate}\parsep=0pt\itemsep=0pt
\item Alice sends decommitment of ${\sf P3}(i,w_i)$ for each $i$
and Bob obtains the committed bits $w_i'$ for all $i$ and computes
$w'=w_1'\oplus \cdots \oplus w_n'$.
\item Bob verifies the correctness of the decommitments.
If all the verification procedures are passed then Bob accepts.
\end{enumerate}

\hrule
\caption{{\sf Protocol 4}}\label{fig:p4}
\end{figure}
%\medskip

%To commit to a bit $w$, we randomly secret share $w=w_1\oplus\cdots\oplus w_m$
%and commit to share $w_i$ using the $i$th commitment scheme.

\begin{theorem}\label{thm:final}
Let $f=\{f_n:\{0,1\}^n \rightarrow \{0,1\}^{n}\}_{n\in\nat}$ be
an $s(n)$-secure quantum one-way function family, where $s(n)=n^{\omega (1)}$.
Then {\sf Protocol 4} with setting of parameters 
$\Delta_1=\Delta_2=\frac{1}{4}\log s(n)$, is a 
computationally-binding and statistically-hiding
quantum bit commitment scheme.
\end{theorem}

Theorem \ref{thm:final} can be also shown as 
Theorems \ref{thm:stronghide} and \ref{thm:parabind}. 
%Theorem \ref{thm:nmfam}.

%\medskip
%
%When we are given an $s(n)$-secure quantum one-way function, 
%we have to know the value $s(n)$ to execute {\sf Protocol 4}, since
%we should specify the hash size for pairwise independent hash functions.
%If we are given a $p$-secure quantum one-way function, we do not have
%to care about the hash size.
%
%\begin{theorem}\label{thm:final2}
%Let $f=\{f_n:\{0,1\}^n \rightarrow \{0,1\}^{n}\}_{n\in\nat}$ be
%a p-secure quantum one-way function family.
%Then {\sf Protocol 4} with setting of parameters 
%$\Delta_1=O(\log n)$ and $\Delta_2=O(\log n)$, is a 
%computationally-binding and statistically-hiding
%quantum bit commitment scheme.
%\end{theorem}
%
%The preceding proofs work enough to show Theorem \ref{thm:final2} 
%by just setting $\Delta_3=\Delta_2+O(\log n)$ in the proof. Note that $\Delta_3$
%appears only in the proof. Thus, we can control the constant factor for $\log n$
%according to the security level of the underlying quantum one-way function.

\section{Concluding Remarks}
We have derived a quantum and non-interactive version (Non-interactive
Quantum Hashing Theorem) of the new interactive
hashing theorem. As its application, we have constructed a statistically-hiding
non-interactive quantum bit commitment scheme. We note that 
by using the same discussion we can show the parallel composability of our
quantum bit commitment scheme.

%Since bit commitment is a fundamental cryptographic primitive, 
%the non-interactivity of our scheme may be used
%to construct 
%lead to reduce the round complexity of 
%bigger quantum cryptographic protocols.
In classical cryptography, the interactive hashing theorem has many applications. 
So, we hope that Non-interactive Quantum Hashing Theorem also has many applications
to quantum cryptography.

\section*{Acknowledgments}
A part of the research was done while the first author was at 
Universit\'e Paris-Sud 11 and supported by JST-CNRS project on
``Quantum Computation: Theory and Feasibility'' and the Ministry
of Education, Science, Sports and Culture, Grant-in-Aid for 
Scientific Research (B), 21300002, 2010 and for Exploratory Research,
20650001, 2010.

%\newpage

\end{document}